\documentclass[11pt,a4paper]{article}

\usepackage[english]{babel}
\usepackage[utf8]{inputenc}
\usepackage[T1]{fontenc}
\usepackage{lmodern}
\usepackage{amsmath,amssymb,amsthm}
\usepackage{geometry}
\usepackage{setspace}
\usepackage{graphicx}
\usepackage{booktabs}
\usepackage{hyperref}
\usepackage{natbib}
\usepackage{xcolor}

\usepackage{tikz}
\usetikzlibrary{arrows.meta,positioning,shapes.geometric}

\newtheorem{definition}{Definition}
\newtheorem{proposition}{Proposition}

\DeclareMathOperator*{\argmin}{arg\,min}
\DeclareMathOperator*{\argmax}{arg\,max}

\usepackage{xcolor}
\definecolor{CiteNavy}{HTML}{0B3D91}
\definecolor{square}{HTML}{876b3a}
\hypersetup{
  colorlinks=true,
  linkcolor=square,
  citecolor=square,
  urlcolor=square
}

\title{Reverse Stress Testing Geopolitical Risk in Corporate Credit Portfolios: A Formal and Operational Framework\thanks{We sincerely thank David Alcaud, Adrien Aubert, Aude Couderc, Cyril Regnier, and all the members of the Square Research Center for their valuable insights and feedback.}}

\author{Christophe Hurlin\thanks{University of Orl\'{e}ans (LEO) and Institut Universitaire de France (IUF), Rue de Blois, 45067 Orl\'{e}ans, France. Corresponding author: christophe.hurlin@univ-orleans.fr}, Quentin Lajaunie\thanks{Square Research Center, Rue des poissonniers, 92220, Neuilly-sur-Seine \& University of Orl\'{e}ans (LEO), Rue de Blois, 45067 Orl\'{e}ans, France.}, Yoann Pull\thanks{Square Research Center, Rue des poissonniers, 92220, Neuilly-sur-Seine \& University of Orl\'{e}ans (LEO), Rue de Blois, 45067 Orl\'{e}ans, France.}}

\date{\today}

\begin{document}

\maketitle

\begin{abstract}
\noindent This paper proposes a formal framework for reverse stress testing geopolitical risk in corporate credit portfolios. A joint macro-financial scenario vector, augmented with an explicit geopolitical risk factor, is mapped into stressed probabilities of default and losses given default. These stresses are then propagated to portfolio tail losses through a latent factor structure and translated into a stressed CET1 ratio, jointly accounting for capital depletion and risk-weighted asset dynamics. Reverse stress testing is formulated as a constrained maximum likelihood problem over the scenario space. This yields a geopolitical point reverse stress test, or design point, defined as the most probable scenario that breaches a prescribed capital adequacy constraint under a reference distribution. The framework further characterises neighbourhoods and near optimal sets of reverse stress scenarios, allowing for sensitivity analysis and governance oriented interpretation. The approach is compatible with internal rating based models and supports implementation at the exposure or sector level.
\end{abstract}

\vspace{0.8cm}

\medskip


\newpage

\section{Introduction}

\begin{quote}
\emph{``The 2026 thematic stress test will take the form of a reverse stress test on geopolitical risk. 
Instead of assessing how a bank performs under a given scenario, reverse stress tests begin by asking what kind of shock would lead to a severe capital depletion.''} \\
--- Claudia Buch, Chair of the ECB Supervisory Board, 
\href{https://www.bankingsupervision.europa.eu/press/blog/2025/html/ssm.blog20250905~781d5f81aa.en.html}{ECB Banking Supervision Blog}, 5 September 2025.
\end{quote}

Reverse stress testing provides a fundamentally different perspective from standard forward-looking supervisory stress tests. Rather than assessing how a bank performs under a given, common adverse scenario, reverse stress tests ask which shocks would lead to a severe depletion of capital in an individual institution. This inversion is particularly valuable in environments characterised by non-linear risk interactions and regime shifts, where severe outcomes may arise from combinations of shocks that are unlikely to be captured by conventional scenario narratives.

This perspective is especially relevant for \emph{geopolitical risk}. The European Central Bank (ECB) identifies geopolitical tensions as a key source of financial vulnerability, highlighting, among other factors, the spike in market volatility following the abrupt announcement of new U.S.\ tariffs in April 2025 and the continued impact of the war in Ukraine \citep{ECB2025fsr}. For banks, geopolitical shocks tend to materialise in a discrete and regime-like manner, propagating simultaneously through credit risk, funding conditions, and market valuations, thereby challenging the design and calibration of standard stress-test scenarios.

Reflecting these characteristics, the ECB announced on 12 December 2025 that it will conduct its first reverse stress test explicitly focused on geopolitical risk in 2026 \citep{ECB_SSM_PR_2025}. The exercise will cover 110 directly supervised banks and will prescribe a capital-based breakdown outcome. Banks will be required to identify relevant geopolitical events that could lead to a depletion of at least 300 basis points in their Common Equity Tier~1 (CET1) capital and to assess not only solvency implications but also potential effects on liquidity and funding conditions.\footnote{To ensure proportionality and cost efficiency, the ECB will run the exercise as part of the 2026 ICAAP data-collection process, with aggregate conclusions expected in summer 2026. In line with previous ECB thematic stress tests, the exercise is not intended to have direct implications for Pillar~2 Guidance; instead, identified weaknesses may feed into the SREP assessment in a qualitative manner. Banks are also expected to describe potential management actions and to demonstrate robust governance and operational resilience frameworks.} The exercise is intended to support supervisory dialogue, governance assessment, and contingency planning, rather than to mechanically determine capital requirements.

Despite this growing policy relevance, existing academic frameworks do not yet provide an integrated treatment of this problem. Quantitative reverse stress testing has been developed primarily for market-risk portfolios and loss-based outcomes \citep{Glasserman2015,grundke2018macroeconomic_rst,Pesenti2019}, while credit stress-testing models are typically forward-looking and focus on the propagation of macroeconomic shocks under given scenarios. Although geopolitical risk is now measurable \citep{caldara2022geopolitical}, it has not been explicitly embedded into a reverse stress testing architecture that links geopolitical and macro-financial shocks to credit risk, portfolio losses, risk-weighted assets, and regulatory capital outcomes.

In this paper, we propose the first formal and operational framework for reverse stress testing geopolitical risk in corporate credit portfolios. Reverse stress testing is defined as the identification of the \emph{most probable} scenario in the admissible scenario space that breaches a prescribed capital adequacy constraint, where scenario plausibility is evaluated through a likelihood function defined relative to a reference distribution for joint geopolitical and macro-financial shocks. Within this framework, geopolitical risk enters explicitly as a component of the scenario vector, and the reverse stress test is formulated as a constrained maximum likelihood problem.

The proposed framework provides a complete and internally coherent mapping from scenarios to regulatory capital outcomes. Joint geopolitical and macro-financial scenarios are translated into stressed probabilities of default (PD) and losses given default (LGD), allowing for heterogeneous sectoral sensitivities. These stressed credit-risk parameters are propagated to portfolio tail losses through a tractable latent-factor dependence structure and mapped into a stressed CET1 ratio by jointly accounting for capital depletion and risk-weighted asset dynamics. As a result, the framework enables the identification and analysis of scenarios consistent with a prescribed capital depletion, including the 300-basis-point reduction in CET1 capital explicitly referenced by the ECB in its geopolitical reverse stress test \citep{ECB_SSM_PR_2025}. 
The resulting \emph{geopolitical point reverse stress scenario}, or \emph{design point}, represents the most plausible joint geopolitical and macro-financial configuration leading to a breach of the capital threshold under the reference distribution.

Beyond the identification of a single breakdown scenario, our framework explicitly accommodates \emph{sets of reverse stress scenarios}, in line with supervisory recommendations \citep{BCBS2018StressTesting,ECB2025StressTestReport}. A local neighbourhood of reverse stress scenarios around the geopolitical design point captures scenario uncertainty in the vicinity of the most probable capital breach. Complementarily, a near-optimal reverse stress scenario set gathers alternative capital-breaching scenarios whose likelihood remains close to that of the design point, even when they correspond to qualitatively different macro-financial narratives. These set-based constructions extend the exercise beyond a single benchmark scenario and support governance-oriented interpretation, sensitivity analysis, and supervisory dialogue.

The remainder of the paper is organised as follows. Section~\ref{Section:Related literature} reviews the related literature. Section~\ref{Section:Scenario space} introduces the scenario space, the geopolitical risk indicator, and the capital framework. Section~\ref{Section:Portfolio losses} develops the mapping from scenarios to portfolio losses and regulatory capital. Section~\ref{sec:rst_optim} formulates reverse stress testing as a constrained optimisation problem and introduces the geopolitical design point and associated scenario sets. Section~\ref{Section:Implementation} discusses implementation issues. Section~\ref{Section:Conclusion} concludes.

\section{Related literature}
\label{Section:Related literature}

This paper relates to three closely connected strands of research: quantitative reverse stress testing, stress testing of banking credit risk, and the measurement and integration of geopolitical risk into macro-financial scenarios. While each strand has developed substantial insights, their intersection remains largely unexplored.

\subsection{Quantitative reverse stress testing}

Reverse stress testing addresses the inverse of conventional stress testing by asking which configurations of risk drivers could plausibly generate a pre-specified adverse outcome. In contrast to forward stress testing, which evaluates the impact of given scenarios, reverse stress testing focuses on the joint severity and plausibility of scenarios leading to failure states. Although this reverse perspective has long been emphasised in supervisory guidance \citep{BCBS2009}\footnote{The Basel Committee on Banking Supervision explicitly refers to reverse stress testing as a complement to conventional stress testing, stating that ``\textit{a stress testing programme should also determine what scenarios could challenge the viability of the bank (reverse stress tests) and thereby uncover hidden risks and interactions among risks}.'' The document further explains that reverse stress tests ``\textit{start from a known stress test outcome, such as a breach of regulatory capital ratios, illiquidity or insolvency, and then ask what events could lead to such an outcome for the bank}'' (\citealp{BCBS2009}, p.~14).}, its formal quantitative implementation has emerged only gradually in the academic literature.

Early contributions emphasise that reverse stress testing can be conducted at different levels of formalisation, ranging from narrative scenario analysis to fully model-based approaches \citep{Grundke2011RST}. In practice, qualitative exercises rely primarily on expert judgement to construct institution-specific breaking scenarios, whereas quantitative approaches embed reverse stress testing within an explicit probabilistic framework. The latter fix a critical loss or failure threshold ex ante and search within a model-generated scenario space for configurations of risk factors that are most strongly associated with such outcomes. In this sense, reverse stress testing inverts the logic of historical or forward stress testing: the adverse outcome is specified first, and model simulations are subsequently mined to identify the most plausible drivers of distress. While such scenarios cannot generally be labelled ex post by reference to historical crises, the underlying models are calibrated to market or macro-financial data that reflect agents’ forward-looking assessments of future risks. As noted by \citet{grundke2018macroeconomic_rst}, this quantitative perspective is essential for moving beyond purely narrative exercises, even though fully developed quantitative reverse stress testing frameworks remain relatively scarce.

A dominant methodological approach defines reverse stress scenarios as the most likely configurations of risk factors conditional on an extreme loss event. In probabilistic terms, reverse stress scenarios correspond to conditional modes of the joint distribution of risk factors given that portfolio losses exceed a prescribed threshold. \citet{Glasserman2015} and \citet{Kopeliovich2015} provide seminal contributions to this approach by formulating reverse stress testing as a constrained optimisation problem. Analytical characterisations are obtained under elliptical distributions and linear loss mappings for tractability, but the underlying principle is more general and centres on identifying the most plausible drivers of severe losses. These results have been extended to richer distributional classes allowing for skewness and heavy tails, notably through skew-elliptical models and multivariate Student-type distributions \citep{Giorgi2016,Schroeder2023}. To relax parametric assumptions while preserving an explicit scenario-based interpretation, \citet{Glasserman2015} propose an empirical likelihood approach in which the reference distribution of risk factors is tilted so as to satisfy a tail loss constraint. 

A distinct but complementary stream shifts the focus from scenario identification to distributional robustness. \citet{Pesenti2019} propose a divergence-based framework for reverse sensitivity testing, in which probability measures are perturbed to the smallest extent necessary to violate a prescribed risk constraint. While this approach provides a powerful diagnostic of model fragility, it places less emphasis on the identification of explicit scenario vectors and may therefore offer more limited interpretability in terms of identifiable economic drivers.

More recently, high-dimensional dependence structures have been explored through copula-based and vine-copula models. These frameworks allow for non-linear dependence, tail dependence, and potentially multiple stress regimes, and can generate several candidate stress scenarios when the conditional distribution is multimodal \citep{zhou2024}. The increased flexibility comes at the cost of greater modelling complexity and calibration requirements, which may limit transparency and operational deployment in governance-oriented stress testing exercises.

Across these contributions, quantitative reverse stress testing has been developed predominantly for market-risk portfolios and loss-based risk measures such as Value-at-Risk or Expected Shortfall. The adverse outcome is typically defined as a portfolio loss exceeding a given threshold, and the mapping from risk factors to losses is often stylised. As a result, existing reverse stress testing frameworks largely abstract from the structural features of banking credit risk and from the regulatory capital mechanics that are central to banking stress testing.

\subsection{Reverse stress testing and banking credit risk}

In the banking context, stress testing primarily focuses on credit risk and capital adequacy. A large body of the literature develops forward-looking stress-testing frameworks that map macroeconomic scenarios into probabilities of default, losses given default, portfolio credit losses, and ultimately regulatory capital ratios (see \cite{Quagliariello2009,HenryKok2013,borio2014stress}, 
among many others). These frameworks underpin supervisory exercises such as CCAR and EBA stress tests and are well established in both academia and industry.\footnote{The Comprehensive Capital Analysis and Review (CCAR) in the United States and the European Banking Authority (EBA) stress tests in Europe are supervisory forward-looking exercises in which regulators prescribe adverse macro-financial scenarios and assess banks’ ability to maintain regulatory capital ratios under stress.}

By contrast, quantitative reverse stress testing has seen limited application to banking credit risk and typically abstracts from PD–LGD dynamics, portfolio credit structures, and regulatory capital mechanics. A small number of studies bring reverse stress testing closer to macroeconomic or banking applications. \citet{grundke2018macroeconomic_rst} formulates reverse stress scenarios at the macroeconomic level, identifying adverse aggregate conditions consistent with severe outcomes, but abstracts from the transmission of shocks through banking credit portfolios and regulatory capital channels. 
Similarly, \citet{Albanese2023} propose bottom-up quantitative reverse stress-testing frameworks for banks’ trading and counterparty exposures. These approaches focus on market-driven P\&L, valuation adjustments, and cost-of-capital metrics, rather than on banking book credit risk and regulatory capital constraints.

Overall, there is currently no formal quantitative reverse stress testing framework that embeds a credit-portfolio model and an explicit regulatory capital constraint within a single probabilistic optimisation problem. This gap is particularly salient given that supervisory definitions of bank failure are inherently capital-based rather than loss-based.

\subsection{Geopolitical risk as a macro-financial scenario driver}

Geopolitical risk has long been treated as a largely narrative and exogenous source of uncertainty, which has limited its formal integration into quantitative stress-testing frameworks. This view has evolved with the development of systematic, time-varying proxies that render geopolitical risk observable and suitable for econometric analysis \citep{caldara2022geopolitical,engle2023events,wei2024evaluating,ferrara_saadaoui2025_measuring_geoeconomics}; see Section~\ref{Subsection: Geopolitical_index}. Text-based indices and event-based measures can now be embedded in standard macro-financial models, allowing geopolitical risk to enter explicitly as a distinct dimension of the scenario space. Beyond academic indicators, geopolitical risk is also captured by proprietary indices developed by major asset managers and financial institutions (e.g.\ BlackRock, Amundi, Bloomberg), as well as by international organisations (e.g.\ the IMF and the World Bank), reflecting its growing role in applied macro-financial analysis and risk management.

A growing empirical literature documents the macro-financial relevance of geopolitical risk. Geopolitical tensions affect real economic activity, financial markets, and banking outcomes through multiple channels. In the banking sector, geopolitical risk is associated with lower bank stability and weaker capital positions \citep{phan2022geopolitical,behn2025insight}, with evidence that adverse effects on capital ratios are concentrated at high levels of geopolitical stress and mitigated by bank size and capital buffers. Beyond balance-sheet effects, geopolitical tensions also influence banks’ intermediation behaviour by raising borrowing costs and tightening nonprice lending conditions \citep{nguyen2023geopolitical}. At a more systemic level, recent cross-country evidence indicates that geopolitical risk amplifies financial stress spillovers and contributes to broader financial fragility by increasing bank risk and fueling asset price imbalances \citep{zhu2025global,wang2025geopolitical}. Taken together, these findings establish geopolitical risk as a systematic macro-financial driver and justify its inclusion alongside traditional macroeconomic variables in stress-testing exercises.

Translating this empirical evidence into operational stress-testing architectures, however, has proven more challenging, and the integration of geopolitical risk into quantitative stress-testing frameworks remains limited. A first attempt is provided by \citet{flament2026}, who incorporate geopolitical risk into forward-looking credit portfolio stress tests within a VAR--Merton framework using Generalized Impulse Response Function (GIRF) analysis. Their approach combines a structural VAR to generate dynamically coherent macro-financial responses to geopolitical shocks with a Merton-type credit satellite, yielding an analytical and internally consistent mapping from geopolitical risk shocks to portfolio default probabilities.
However, to our knowledge, no reverse stress-testing framework explicitly incorporates geopolitical risk, despite its identification as a supervisory priority by the ECB. More broadly, the inverse problem of identifying the most plausible combinations of geopolitical and macro-financial shocks that lead to a severe capital depletion outcome remains unaddressed.

\section{Scenario space, geopolitical risk, and capital framework}
\label{Section:Scenario space}

\subsection{Portfolio, capital structure, and breakdown threshold}

Consider a bank with a corporate credit portfolio composed of exposures indexed by $i = 1,\dots,n$. For each exposure $i$, let $\text{EAD}_i$ denote the exposure at default, $\text{PD}_i^0$ the current probability of default over the stress horizon, and $\text{LGD}_i^0$ the loss given default. The superscript $0$ refers to current, unstressed values. The total exposure is given by
\begin{equation}
  \text{EAD} = \sum_{i=1}^n \text{EAD}_i .
\end{equation}

The bank has Common Equity Tier~1 capital $\text{CET1}_0$ and risk weighted assets $\text{RWA}_0$\footnote{Common Equity Tier~1 (CET1) capital corresponds to the highest quality component of regulatory capital, composed mainly of common equity and retained earnings.}. 
The corresponding CET1 ratio is
\begin{equation}
  R_0 = \frac{\text{CET1}_0}{\text{RWA}_0}.
\end{equation}
Reverse stress testing is conducted relative to a breakdown threshold for the CET1 ratio, denoted by $R^{\star}$, with $R^{\star}>0$ and 
\begin{equation}
  R^{\star} < R_0 .
\end{equation}
The objective of the reverse stress test is to identify plausible geopolitical and macro financial scenarios under which the CET1 ratio stressed at horizon $h$, denoted by $R(s)$, falls below the breakdown threshold $R^{\star}$.

The choice of the threshold $R^{\star}$ is a key modelling input. In the present framework, $R^{\star}$ represents the level of capitalisation below which the bank is considered to experience severe capital distress, and its specification depends on the purpose of the exercise. In its recent announcement of a geopolitical risk reverse stress test, the ECB specifies that banks will be asked “\emph{to identify the most relevant geopolitical risk events that could lead to at least a 300-basis point depletion in }[their] \emph{Common Equity Tier~1 (CET1) capital}”. In line with this supervisory benchmark, we define $R^{\star}$ in relative terms as a fixed depletion of the initial CET1 ratio, namely
\begin{equation}
  R^{\star} = R_0 \times (1 - \Delta),
\end{equation}
with $\Delta = 0.03$.
This definition provides a transparent ratio based counterpart to a 300 basis point CET1 depletion and allows the reverse stress test to be formulated directly in terms of the CET1 ratio used in supervisory assessments.\footnote{Alternatively, the breakdown threshold can be defined as $\text{CET1}^{\star} = (1-0.03)\times\text{CET1}_0$, corresponding to a 300 basis point depletion of CET1 capital. This formulation is equivalent to the ratio based definition when risk weighted assets remain approximately stable over the stress horizon.}

An alternative is to anchor $R^{\star}$ to regulatory requirements by setting it equal to the Pillar~1 minimum CET1 ratio (4.5 per cent; \citealp{BCBS_RBC20}) plus the applicable capital buffers (e.g. the capital conservation buffer; \citealp{BCBS_RBC30}). Under this specification, the breakdown condition corresponds to a breach of the combined buffer requirement and directly aligns the reverse stress test with supervisory solvency thresholds. Both definitions are compatible with the proposed framework and yield complementary interpretations of the resulting reverse stress scenarios.

\subsection{Scenario vector and geopolitical risk indicator}

Let $x$ denote a $(d-1)$-dimensional vector of macro-financial shocks, interpreted as a realisation of a random vector $X \in \mathcal{X} \subseteq \mathbb{R}^{d-1}$. The components of $x$ may capture changes in real activity, interest rates, credit spreads, commodity prices, or sector-specific demand conditions. Let $g$ denote the geopolitical risk component, interpreted as a realisation of a scalar random variable $G \in \mathcal{G} \subseteq \mathbb{R}$. The variable $g$ is normalised so that $g=0$ corresponds to the current geopolitical environment, while positive values indicate an intensification of geopolitical risk. In the theoretical analysis, $G$ is treated as an abstract exogenous random variable capturing geopolitical shocks. In the empirical implementation, it is proxied by a geopolitical risk index; see Section~\ref{Subsection: Geopolitical_index} for details. The full scenario vector $s$ is defined as
\begin{equation}
  s =
  \begin{pmatrix}
    g \\
    x
  \end{pmatrix}
  \in \mathcal{S},
\end{equation}
where $\mathcal{S} := \mathcal{G} \times \mathcal{X} \subseteq \mathbb{R}^{d}$ denotes the admissible scenario set, and
$S = (G,X)^\top$ is the associated $d$-dimensional random vector.

The joint distribution of $S$ is denoted by $P_S(\,\cdot\,\vert\mathcal{I})$, where $\mathcal{I}$ represents the information set available at horizon $h$, at which the reverse stress test is conducted.\footnote{To simplify notation, we abstract from the explicit time dimension. Formally, the scenario vector $s$ should be indexed at time $t=h$, where $t=0$ denotes the date at which the reverse stress test is initiated, and the information set $\mathcal{I}$ should be indexed at time $t=0$.} In the reverse stress testing exercise, the optimisation is carried out over scenario realisations $s \in \mathcal{S}$, while the plausibility of such scenarios is assessed relative to the reference distribution $P_S(\,\cdot\,\vert\mathcal{I})$. This probability measure is defined on the Borel $\sigma$-algebra of $\mathbb{R}^d$ and captures both the marginal behaviour of geopolitical and macro-financial shocks and their dependence structure. The formulation is deliberately general and encompasses a wide class of modelling approaches. The distribution $P_S(\,\cdot\,\vert\mathcal{I})$ may be specified parametrically or non-parametrically, constructed via a copula linking the marginal distributions of $G$ and $X$, or generated by a dynamic multivariate model conditional on $\mathcal{I}$, such as a vector autoregression model (VAR) capturing the joint dynamics of geopolitical risk indices and macro-financial variables, a state-space model, or a regime-switching process.

\section{Portfolio losses and capital under geopolitical stress scenarios}
\label{Section:Portfolio losses}

\subsection{Portfolio loss distribution under a geopolitical scenario}

This subsection develops a structural mapping from geopolitical and macro financial scenarios to the portfolio loss distribution, by modelling their joint impact on probabilities of default, losses given default, default dependence, and ultimately on tail portfolio losses through an analytically tractable factor model.

\newpage

\paragraph{Mapping scenarios and geopolitical risk to probabilities of default.} For each exposure $i$ we specify a functional relationship between the scenario vector and the stressed probability of default
\begin{equation}
  \text{PD}_i(s) = \text{PD}_i(g,x) = f_i(g,x),
\end{equation}
where $f_i$ is constrained to produce values in the open interval $(0,1)$. A parsimonious specification can be obtained by grouping exposures into sectors, indexed by $k = 1,\dots,K$, and by assuming that the sensitivity of the log odds of default to geopolitical risk and to macro financial variables is constant within each sector.

Let $k(i)$ denote the sector to which exposure $i$ belongs and let $\beta_{k(i)}$ and $\delta_{k(i)}$ be sector specific sensitivity vectors for macro financial and geopolitical shocks respectively. The logit specification is
\begin{equation}
  \log\left(\frac{\text{PD}_i(g,x)}{1-\text{PD}_i(g,x)}\right)
  =
  \log\left(\frac{\text{PD}_i^0}{1-\text{PD}_i^0}\right)
  + \beta_{k(i)}^{\top} x
  + \delta_{k(i)} g.
\end{equation}
Solving for $\text{PD}_i(g,x)$ yields
\begin{equation}
  \text{PD}_i(g,x)
  =
  \frac{\text{PD}_i^0
  \exp\big(\beta_{k(i)}^{\top} x + \delta_{k(i)} g\big)}
  {1 - \text{PD}_i^0 + \text{PD}_i^0 \exp\big(\beta_{k(i)}^{\top} x + \delta_{k(i)} g\big)}.
\end{equation}
Positive values of $\delta_{k(i)}$ reflect the idea that an increase in geopolitical risk raises the probability of default in sectors that are more exposed to global trade, supply chain disruptions or sanctions. 
Negative values could capture beneficiary sectors such as defence industries, although in a stress context it is natural to impose sign constraints on $\delta_{k(i)}$.

\vspace{-0.5cm}
\paragraph{Mapping scenarios and geopolitical risk to loss given default.}
The impact of geopolitical and macro financial shocks on loss given default is represented by
\begin{equation}
  \text{LGD}_i(s)=\text{LGD}_i(g,x) = g_i(g,x),
\end{equation}
where $g_i$ produces values in $(0,1)$. 
A convenient specification is an affine function in $g$ and $x$:
\begin{equation}
  \text{LGD}_i(g,x)
  =
  \text{LGD}_i^0 + \gamma_{k(i)}^{\top} x + \eta_{k(i)} g,
\end{equation}
where $\gamma_{k(i)}$ and $\eta_{k(i)}$ are sector specific sensitivity parameters, and $\text{LGD}_i(g,x)$ is truncated to the unit interval if necessary. 
The term $\eta_{k(i)} G$ captures the effect of geopolitical risk on recoveries through channels such as asset seizures, disruptions in cross border enforcement, collateral value losses or restrictions on the sale of pledged assets.

\vspace{-0.5cm}
\paragraph{Latent factor model for defaults.}
Conditional on the scenario $(g,x)$, default dependence among obligors is modelled through a one factor Gaussian latent variable structure, as in the standard ASRF/IRB framework \citep{vasicek2002distribution,Gordy2003}. 
For each exposure $i$ define the latent variable
\begin{equation}
  Y_i = \sqrt{\rho_i}\, Z + \sqrt{1 - \rho_i}\, \varepsilon_i,
\end{equation}
where $Z$ is a standard normal systematic factor, $\varepsilon_i$ is an idiosyncratic standard normal variable, and $\rho_i \in (0,1)$ is an asset correlation parameter. 
The random variables $Z$ and $\varepsilon_i$ are independent, and the $\varepsilon_i$ are independent across obligors. Default occurs if
\begin{equation}
  Y_i \leq \Phi^{-1}\big(\text{PD}_i(g,x)\big),
\end{equation}
where $\Phi(.)$ denotes the standard normal cumulative distribution function. 
The default indicator for exposure $i$ is then
\begin{equation}
  D_i(g,x) = \mathbf{1}\left\{Y_i \leq \Phi^{-1}\big(\text{PD}_i(g,x)\big)\right\}.
\end{equation}
The conditional loss on exposure $i$ under scenario $(g,x)$ is
\begin{equation}
  L_i(g,x) = \text{EAD}_i \, \text{LGD}_i(g,x)\, D_i(g,x),
\end{equation}
and the total portfolio loss is
\begin{equation}
  L(g,x) = \sum_{i=1}^n L_i(g,x).
\end{equation}

\vspace{-0.5cm}
\paragraph{Analytical approximation of tail losses.}
In the asymptotic homogeneous portfolio approximation, the distribution of portfolio losses conditional on $Z$ admits a tractable expression \citep{Gordy2003}. 
An analytical approximation for the quantile of the loss distribution at confidence level $q$ is
\begin{equation}
  L_q(g,x) \approx \sum_{i=1}^n \text{EAD}_i \,\text{LGD}_i(g,x)\,
  \Phi\left(
  \frac{\Phi^{-1}\big(\text{PD}_i(g,x)\big)
  + \sqrt{\rho_i}\,\Phi^{-1}(q)}
  {\sqrt{1-\rho_i}}
  \right).
\end{equation}
This expression links the geopolitical and macro financial scenario $(g,x)$ to a portfolio tail loss measure that can be compared with available capital.

\subsection{Capital and risk weighted assets under geopolitical stress}

We now translate the portfolio loss distribution derived above into stressed CET1 ratio, thereby linking geopolitical and macro financial scenarios to the bank’s capital adequacy condition.

\vspace{-0.5cm}
\paragraph{CET1 capital under stress.} Under a given scenario $(g,x)$, CET1 capital at the stress horizon $h$ is approximated by
\begin{equation}
    \label{eq:CET1}
  \text{CET1}(g,x) = \text{CET1}_0 - L_q(g,x) 
  + \Delta \text{P\&L}_{\text{non credit}}(g,x),
\end{equation}
where $\Delta \text{P\&L}_{\text{non credit}}(g,x)$ represents the contribution of non credit activities. 
Geopolitical risk can affect these additional terms through trading losses, fee income reductions or higher operating costs. In a pure credit portfolio analysis these terms can be set to zero or modelled through simple linear approximations. 

\vspace{-0.5cm}
\paragraph{Risk weighted assets under stress.}
Risk weighted assets may react to geopolitical scenarios through changes in probabilities of default and loss given default. At asset level, the risk weight of exposure $i$ under scenario $(g,x)$ is denoted by $\text{RW}_i(g,x)$ and can be written as
\begin{equation}
  \text{RW}_i(g,x) = h\big(\text{PD}_i(g,x), \text{LGD}_i(g,x), \text{M}_i\big),
\end{equation}
where $M_i$ is the effective maturity and the function $h(.)$ is defined by
\begin{equation}
  h\big(\text{PD}, \text{LGD}, \text{M}\big)
  =
  \text{LGD} \times 
  \left[
  \Phi\left(
  \frac{\Phi^{-1}\big(\text{PD}\big)
  + \sqrt{\rho}\,\Phi^{-1}(q)}
  {\sqrt{1-\rho}}
  \right)
  - \text{PD}
  \right] \times \gamma(\text{M}),
\end{equation}
with $\rho$ and $q$ denoting, respectively, the asset correlation parameter and the confidence level used to compute the loss quantile, and where $\gamma(\text{M})$ represents the maturity adjustment \citep{BCBS_CRE31}. The total risk weighted assets are
\begin{equation}
  \text{RWA}(g,x) = \sum_{i=1}^n \text{EAD}_i \, \text{RW}_i(g,x).
\end{equation}
This relationship can be replaced by a linear approximation
\begin{equation}
    \label{eq:RWA}
  \text{RWA}(g,x) = \text{RWA}_0 + \sum_{i=1}^n \alpha_i \big(\text{PD}_i(g,x)-\text{PD}_i^0\big),
\end{equation}
where the coefficients $\alpha_i$ are calibrated to internal capital model outputs and historical observations. 
A positive level of geopolitical risk $G$ then contributes to higher risk weights through its impact on probabilities of default.

\vspace{-0.5cm}
\paragraph{CET1 ratio under stress.}
The CET1 ratio under a geopolitical scenario $s=(g,x)^{\top}$ is
\begin{equation}
  R(s) = \frac{\text{CET1}(g,x)}{\text{RWA}(g,x)},
\end{equation}
where $\text{CET1}(g,x)$ and $\text{RWA}(g,x)$ are defined in
equations \eqref{eq:CET1} and \eqref{eq:RWA}, respectively.
The bank is considered to experience a capital breakdown whenever the
stressed CET1 ratio falls below the breakdown threshold $R^{\star}$, that is,
\begin{equation}
  R(s) \leq R^{\star}.
\end{equation}

\section{Reverse stress testing as constrained optimisation}
\label{sec:rst_optim}

\subsection{Geopolitical point reverse stress test}
\label{sec:point_RST}

In its press release of 12 December 2025, the ECB states that ''\textit{in a reverse stress test, a pre-determined outcome is prescribed and each bank defines the scenario in which that outcome would materialise}'' \citep{ECB_SSM_PR_2025}. Consistent with this principle, we formulate the reverse stress testing problem as the identification of a geopolitical and macro-financial scenario $s$ for which the stressed CET1 ratio satisfies $R(s) \leq R^{\star}$ and that is as plausible as possible under the joint distribution for geopolitical and macro-financial risks.

Let $f(s \vert \mathcal{I})$ denote the conditional probability density function associated with the distribution $P_S(\,\cdot\,\vert\mathcal{I})$ of the scenario vector $S$ at horizon $h$, where $P_S(\,\cdot\,\vert\mathcal{I})$ is defined conditional on the information set $\mathcal{I}$ available at the initial date at which the reverse stress test is conducted. A point reverse stress test can then be formulated as a constrained maximum likelihood problem over the capital-breaching region 
\citep{Glasserman2015,Pesenti2019}. 

\begin{definition}[\textbf{Geopolitical point reverse stress test}]
    \label{Def: Geopolitical point RST}
    A geopolitical point reverse stress test consists of solving
    \begin{equation}
      s^{\star}
      \in
      \argmax_{s \in \mathcal{S}}
      f(s \vert \mathcal{I})
      \quad \text{s.t.} \quad
      R(s) \leq R^{\star},\; g > 0.
    \end{equation}
    The solution $s^\star$, referred to as the point reverse stress scenario or design point, represents the most likely joint geopolitical and macro-financial configuration under the reference distribution for which the CET1 ratio falls below the prescribed breakdown threshold $R^{\star}$. 
\end{definition}
\noindent The qualifier \emph{point} emphasises that the optimisation problem admits a unique solution. The term \emph{design point} is borrowed from the structural reliability literature, where it denotes the most probable point (MPP) of failure, defined as the point on the failure boundary that minimises the distance to the origin in the space of standardised inputs. The geopolitical point reverse stress test introduced in this paper is the direct analogue of this concept, with the capital breakdown constraint playing the role of the failure domain.

For convenience, we define the \emph{reverse stress breach scenario set}, denoted by $\mathcal{S}_{\mathrm{red}}$, as the set of scenarios that breach the capital constraint (the \emph{``red zone''}), with 
\begin{equation}
  \mathcal{S}_{\mathrm{red}}
  =
  \left\{
    s \in \mathcal{S} : R(s) \le R^{\star}
  \right\}.
\end{equation}

The condition $g > 0$ in the constrained optimisation problem imposes a deterioration in geopolitical conditions relative to the baseline. The optimisation problem can further be augmented with additional feasibility constraints, such as bounds on the admissible scenario components,
\begin{equation}
  g_{\min} \leq g \leq g_{\max},
  \qquad
  x_{\min} \leq x \leq x_{\max},
\end{equation}
as well as monotonicity constraints ruling out improvements in credit quality under stress, for instance
\begin{equation}
  \text{PD}_i(s) \geq \text{PD}_i^{0},
  \qquad
  \text{LGD}_i(s) \geq \text{LGD}_i^{0},
  \qquad i = 1,\dots,n.
\end{equation}

Under mild regularity conditions, most notably continuity of $R(\cdot)$ and the existence of at least one capital-breaching scenario, the capital constraint is binding at the optimum (see \citealp[Sec.~5.5]{boyd2004convex}). Indeed, suppose that a scenario $s$ satisfies $R(s) < R^{\star}$. Then, for any $\alpha \in (0,1)$, the scaled scenario $\alpha s$ is more plausible under an elliptical reference model and, by continuity of $R(\cdot)$, remains capital-breaching for $\alpha$ sufficiently close to $1$. As a result, an interior feasible point cannot be optimal.

It follows that the solution $s^{\star} = (g^{\star}, x^{\star})$ lies on the breakdown frontier
\begin{equation}
\partial \mathcal{S}_{\mathrm{red}} = \{ s \in \mathcal{S} : R(s) = R^{\star} \}.
\end{equation}
The optimal scenario is therefore not an arbitrary point on the boundary, but the one that maximises the likelihood under the reference distribution among all scenarios that exactly exhaust the capital constraint.

\paragraph{Normality assumption.}
Definition~\ref{Def: Geopolitical point RST} can be operationalised under any reference distribution for the scenario vector $S=(G,X)^{\top}$; see, for instance, Appendix~\ref{Appendix:student_reference} for the case of a Student distribution. In what follows, we assume a multivariate normal reference distribution,
\begin{equation}
  S \sim \mathcal{N}(0,\Sigma),
\end{equation}
where $\Sigma$ denotes the conditional covariance matrix of joint geopolitical and macro-financial shocks. Since the Gaussian density admits the representation
\begin{equation}
  f(s\vert\mathcal{I})
  \propto
  \exp\!\left(-\tfrac{1}{2} s^{\top}\Sigma^{-1}s\right),
\end{equation}
the likelihood-based formulation of the point reverse stress test is equivalent to the minimisation of a quadratic form in the scenario realisation $s$.

\begin{definition}[\textbf{Geopolitical point reverse stress test under normality}]
    \label{Def: point RST Normality}
    Under the normal reference model, the geopolitical point reverse stress scenario is defined as the solution of
    \begin{equation}
      s^{\star} \in \argmin_{s \in \mathcal{S}_{\mathrm{red}},\, g>0} \frac{1}{2} d_{\Sigma}^2(s).
    \end{equation}
    where $d_{\Sigma}^{2}(s)=s^{\top}\Sigma^{-1}s$ denotes the squared Mahalanobis distance associated with the covariance matrix $\Sigma$.
\end{definition}
\noindent Under the normality assumption, scenario plausibility is fully characterised by the squared Mahalanobis distance to the baseline, which provides a scale-free measure of the extremeness of the optimal scenario $s^{\star}$ under the reference model. 

\vspace{-0.5cm}
\paragraph{Plausibility score and chi-squared calibration.} In the Gaussian case, the squared Mahalanobis distance follows a $\chi^{2}_{d}$ distribution,
\begin{equation}
  d_{\Sigma}^{2}(S) \sim \chi^{2}_{d}.
\end{equation}
As a result, a natural plausibility measure of the scenario $s^\star$ is given by the tail probability
\begin{equation}
  \mathbb{P}\!\left( d_{\Sigma}^{2}(S) \ge d_{\Sigma}^{2}(s^{\star}) \right)
  =
  1 - F_{\chi^{2}_{d}}\!\left( d_{\Sigma}^{2}(s^{\star}) \right),
\end{equation}
where $\mathbb{P}(\cdot) := P_S(\cdot \mid \mathcal{I})$ and $F_{\chi^{2}_{d}}$ denotes the cumulative distribution function of the $\chi^{2}_{d}$ distribution. This quantity can be interpreted as a p-value--like measure: it is the probability of observing a joint shock at least as extreme as $s^{\star}$ in the Mahalanobis metric. Because this measure is invariant to linear rescalings of the risk drivers, it offers a coherent basis for comparing scenarios across portfolios and model specifications, and for benchmarking their extremeness against historical episodes associated with major geopolitical disruptions.

\newpage
\subsection{Multiple plausible scenarios}
\label{sec:multiple_scenarios}

A single design point $s^\star$ provides a clear and interpretable anchor, but stress testing is foremost
a governance tool. Supervisory principles therefore emphasise that an effective stress-testing programme
should span a \emph{range} of scenarios and severity levels to
reduce the risk of overlooking relevant vulnerabilities \citep{BCBS2018StressTesting}. Consistent with
this perspective, the ECB notes that ''\textit{multiple stress scenarios \dots\ can expose vulnerabilities that would not be detected if a single scenario approach were taken}'' \citep[p.~10]{ECB2025StressTestReport}.
Accordingly, we complement the point reverse stress test with \emph{sets} of plausible capital-breaching scenarios.

We consider two complementary constructions. The first is a \emph{local neighbourhood of reverse stress scenarios} around the point reverse stress scenario $s^\star$, which characterises scenario uncertainty in the immediate vicinity of the design point. The second is a \emph{set of plausible reverse stress scenarios}, defined as a near-optimal relaxation of the likelihood maximisation programme. This latter construction captures alternative capital-breaching narratives that remain almost as plausible as $s^\star$, and is particularly relevant when the breakdown frontier is non-convex. 
An operational procedure to map these set-valued outcomes into a finite, reportable scenario list is presented in Section~\ref{subsec:scenario_list}.

\vspace{-0.5cm}
\paragraph{Local neighbourhood of reverse stress scenarios.} 
To explore perturbations \emph{around} the point reverse stress scenario $s^\star$, we introduce a local neighbourhood
that captures nearby capital-breaching scenarios in a model-consistent notion of proximity. 

\begin{definition}[\textbf{Local neighbourhood of reverse stress scenarios}]
\label{def:local_neighbourhood}
    Let $\delta:\mathcal{S}\times\mathcal{S}\to\mathbb{R}_+$ be a distance function capturing a model-consistent notion of proximity between scenarios. For $\eta>0$, define the $\delta$-ball around $s^\star$ as
    \begin{equation}
      \mathcal{B}^{\delta}_{\eta}(s^\star)
      :=
      \big\{\, s\in\mathcal{S}:\ \delta(s,s^\star)\le \eta \big\},
    \end{equation}
    and define the associated local neighbourhood of reverse stress scenarios as
    \begin{equation}
      \mathcal{S}^{\delta}_{\eta}
      :=
      \mathcal{S}_{\mathrm{red}}\cap \mathcal{B}^{\delta}_{\eta}(s^\star).
    \end{equation}
\end{definition}
\noindent When there is no ambiguity on the choice of $\delta$, we simply write $\mathcal{B}_{\eta}(s^\star)$ and $\mathcal{S}_{\eta}$ for $\mathcal{B}^{\delta}_{\eta}(s^\star)$ and $\mathcal{S}^{\delta}_{\eta}$, respectively. A convenient way to construct such a $\delta$ is to consider a norm $\|\cdot\|$ or a semi-norm $p$, on $\mathbb{R}^d$ and set $\delta(s,t)=\|s-t\|$ (or $\delta(s,t)=p(s-t)$).\footnote{A semi-norm $p:\mathbb{R}^d\to\mathbb{R}_+$ is positively homogeneous and subadditive, i.e. $p(\alpha u)=|\alpha|\,p(u)$ and $p(u+v)\le p(u)+p(v)$, but unlike a norm it may satisfy $p(u)=0$ for some $u\neq 0$ (e.g.\ $p(u)=\|Au\|_2$ when $A$ is not full rank).} 

Under the Gaussian reference model, proximity can naturally be measured using the Mahalanobis geometry induced by $\Sigma$, leading to the Mahalanobis ball illustrated in Figure~\ref{fig:rst_geometry}.

\begin{proposition}[\textbf{Local neighbourhood of scenarios under normality}]
\label{prop:local_ball_normality}
    Assume that $S\sim\mathcal{N}(0,\Sigma)$ and let $d_\Sigma^2(s)=s^\top\Sigma^{-1}s$. Consider the choice
    $\delta(s,s^\star)=d_\Sigma(s-s^\star)$. Then
    \begin{equation}
      \mathcal{B}^{\delta}_{\eta}(s^\star)
      =
      \big\{\, s\in\mathcal{S}:\ d_\Sigma^2(s-s^\star)\le \eta \big\},
      \qquad
      \mathcal{S}^{\delta}_{\eta}
      =
      \mathcal{S}_{\mathrm{red}}\cap
      \big\{\, s:\ d_\Sigma^2(s-s^\star)\le \eta \big\}.
    \end{equation}
\end{proposition}

\begin{proof}
The result follows by direct substitution of $\delta(s,s^\star)=d_\Sigma(s-s^\star)$ into the definitions of
$\mathcal{B}^{\delta}_{\eta}(s^\star)$ and $\mathcal{S}^{\delta}_{\eta}$.
\end{proof}

Figure~\ref{fig:rst_geometry} provides a geometric interpretation of the local neighbourhood of reverse stress scenarios in the $(g,x)$ plane, in the case of a single macroeconomic risk factor ($d=2$). Elliptic contours represent levels of equal plausibility under the reference distribution, while the solid curve denotes the breakdown frontier
$\partial \mathcal{S}_{\mathrm{red}}= \left\{s \in \mathcal{S} : R(s) = R^{\star}\right\}$. The benchmark reverse stress scenario $s^\star$ corresponds to the tangency point between the smallest plausibility contour and the frontier, making it the most plausible capital-breaching configuration. The local neighbourhood $\mathcal{S}_\eta$ is obtained by intersecting the reverse stress event set with a plausibility ball centred at $s^\star$, and collects nearby reverse stress scenarios corresponding to small perturbations of the point reverse stress scenario.
\vspace{0.5cm}

\begin{figure}[!htbp]
  \centering
  \includegraphics[width=0.85\linewidth]{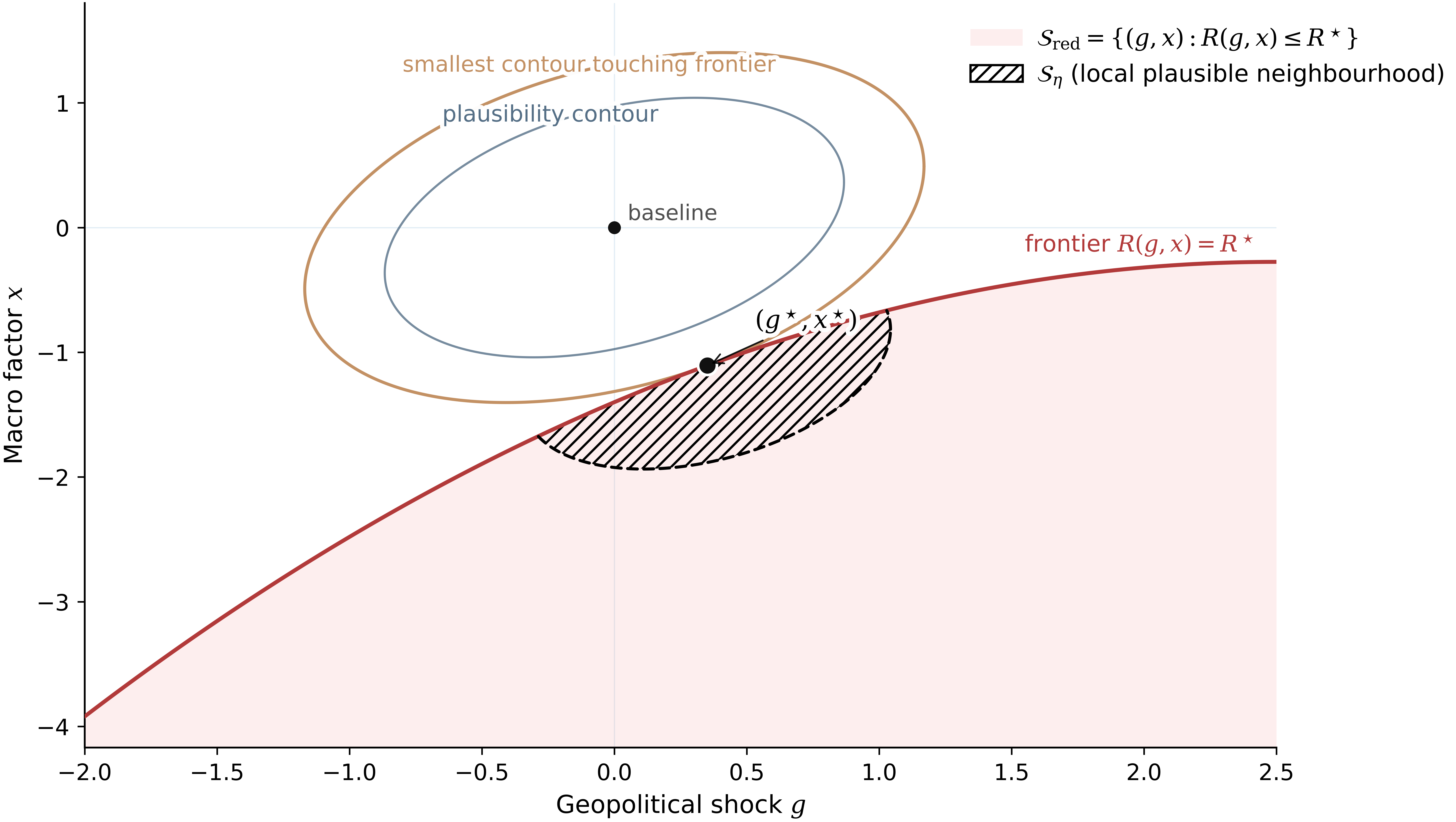}
  \caption{Geometry of the reverse-stress problem in the $(g,x)$ plane. The shaded region is the reverse stress breach scenario set (red zone) $\mathcal{S}_{\mathrm{red}}=\{(g,x):R(g,x)\le R^\star\}$ and the solid curve is the breakdown frontier $R(g,x)=R^\star$. Elliptic contours centred at $s=0$ are level sets of $d_\Sigma^{2}(s)=s^\top\Sigma^{-1}s$ under the reference distribution. The point reverse stress test scenario $s^\star=(g^\star,x^\star)$ is the tangency point between the frontier and the smallest plausibility contour, obtained under the constraint $g\ge0$. The dashed ellipsoid represents the local ball $\mathcal{B}_\eta(s^\star)$ and the local neighbourhood of reverse stress scenarios $\mathcal{S}_\eta=\mathcal{S}_{\mathrm{red}}\cap\mathcal{B}_\eta(s^\star)$.}
  \label{fig:rst_geometry}
\end{figure}

\vspace{-0.5cm}
\paragraph{Plausible reverse stress scenario set.}
The local neighbourhood $\mathcal{S}_\eta$ is designed to explore perturbations \emph{around} the benchmark scenario
$s^\star$. When the breakdown frontier is non-convex, however, alternative crisis narratives may exist that are distant
from $s^\star$ in the state space while remaining almost as plausible under the reference distribution. 

To capture such alternatives, we introduce a \emph{plausible reverse stress scenario set} defined as an $\varepsilon$-relaxation of the likelihood maximisation problem in Definition~\ref{Def: Geopolitical point RST}. This construction collects all capital-breaching scenarios whose plausibility under the reference distribution lies within an
$\varepsilon$-neighbourhood of the benchmark design point.

\newpage
\begin{definition}[\textbf{Plausible reverse stress scenario set}]
    \label{def:set_rst_general}
    Let $s^\star$ denote the solution to the geopolitical point reverse stress test defined in
    Definition~\ref{Def: Geopolitical point RST}. For a tolerance level $\varepsilon>0$, define the
    $\varepsilon$-near-optimal reverse stress scenario set as
    \begin{equation}
      \mathcal{N}_\varepsilon
      :=
      \Big\{\, s\in \mathcal{S} \;:\; R(s)\le R^\star,\; g>0,\;
      -\log f(s\mid\mathcal{I}) \le -\log f(s^\star\mid\mathcal{I}) + \varepsilon/2 \Big\}.
      \label{eq:N_eps_def}
    \end{equation}
\end{definition}
\noindent Definition~\ref{def:set_rst_general} collects all capital-breaching scenarios whose negative log-likelihood remains within $\varepsilon/2$ of the optimum. Equivalently, $\mathcal{N}_\varepsilon$ can be written as the set of plausible geopolitical reverse stress scenarios whose likelihood ratio relative to the design point exceeds $e^{-\varepsilon/2}$, that is, $f(s\mid\mathcal{I}) \ge e^{-\varepsilon/2} f(s^\star\mid\mathcal{I})$. This relaxation is consistent with the conceptual goals of reverse stress testing emphasised by \citet{Kopeliovich2015}: scenarios should be (relatively) likely, sufficiently
distinct, and should not exclude severe outcomes.

Under the Gaussian reference model, $-\log f(s\mid\mathcal{I})$ is proportional to the squared Mahalanobis  distance $d_\Sigma^2(s)$, and Definition~\ref{def:set_rst_general} admits a simple geometric characterisation.

\begin{proposition}[\textbf{Plausible reverse stress scenario set under normality}]
    \label{prop:near_optimal_normality}
    Assume that the reference model for $S$ is Gaussian, $S \sim \mathcal{N}(0,\Sigma)$. Then the set of plausible reverse stress scenarios admits the following equivalent geometric characterisation:
    \begin{equation}
      \mathcal{N}_\varepsilon
      =
      \Big\{\, s \in \mathcal{S} : R(s) \le R^\star,\; g > 0,\;
      d_\Sigma^2(s) \le d_\Sigma^2(s^\star) + \varepsilon \Big\}.
    \end{equation}
\end{proposition}
\begin{proof}
    Under the Gaussian reference model $S\sim\mathcal{N}(0,\Sigma)$, there exists a constant $C>0$ such that
    $f(s\mid\mathcal{I}) = C \exp\!\big(-\tfrac{1}{2} d_\Sigma^2(s)\big)$.
    Taking minus logarithms yields
    $-\log f(s\mid\mathcal{I}) = \tfrac{1}{2} d_\Sigma^2(s) - \log C$.
    Therefore,
    \begin{equation*}
      -\log f(s\mid\mathcal{I}) \le -\log f(s^\star\mid\mathcal{I}) + \tfrac{\varepsilon}{2}
      \iff
      \tfrac{1}{2}\!\left(d_\Sigma^2(s)-d_\Sigma^2(s^\star)\right)\le \tfrac{\varepsilon}{2}
      \iff
      d_\Sigma^2(s)\le d_\Sigma^2(s^\star)+\varepsilon,
    \end{equation*}
    which proves the claim.
\end{proof}

\noindent Figure~\ref{fig:rst_near_optimal} provides a geometric interpretation of the
plausible reverse stress scenario set $\mathcal{N}_\varepsilon$ in the $(g,x)$ plane, again in the
case of a single macroeconomic risk factor ($d=2$). As in Figure~\ref{fig:rst_geometry}, elliptic contours
represent levels of equal plausibility under the reference distribution, while the solid curve denotes the
breakdown frontier $\partial \mathcal{S}_{\mathrm{red}}=\{s\in\mathcal{S}:R(s)=R^\star\}$. The 
point reverse stress scenario $s^\star$ corresponds to the tangency point between the smallest plausibility contour
and the frontier. The outer contour represents the relaxed plausibility level
$d_\Sigma^2(s)=d_\Sigma^2(s^\star)+\varepsilon$. The hatched region $\mathcal{N}_\varepsilon$, obtained by
intersecting this plausibility sublevel set with the reverse stress breach scenario set, illustrates that
near-optimal capital-breaching scenarios may arise at locations that are distant from $s^\star$ in the state
space when the breakdown frontier is non-convex.

\begin{figure}[!htbp]
  \centering
  \includegraphics[width=0.85\linewidth]{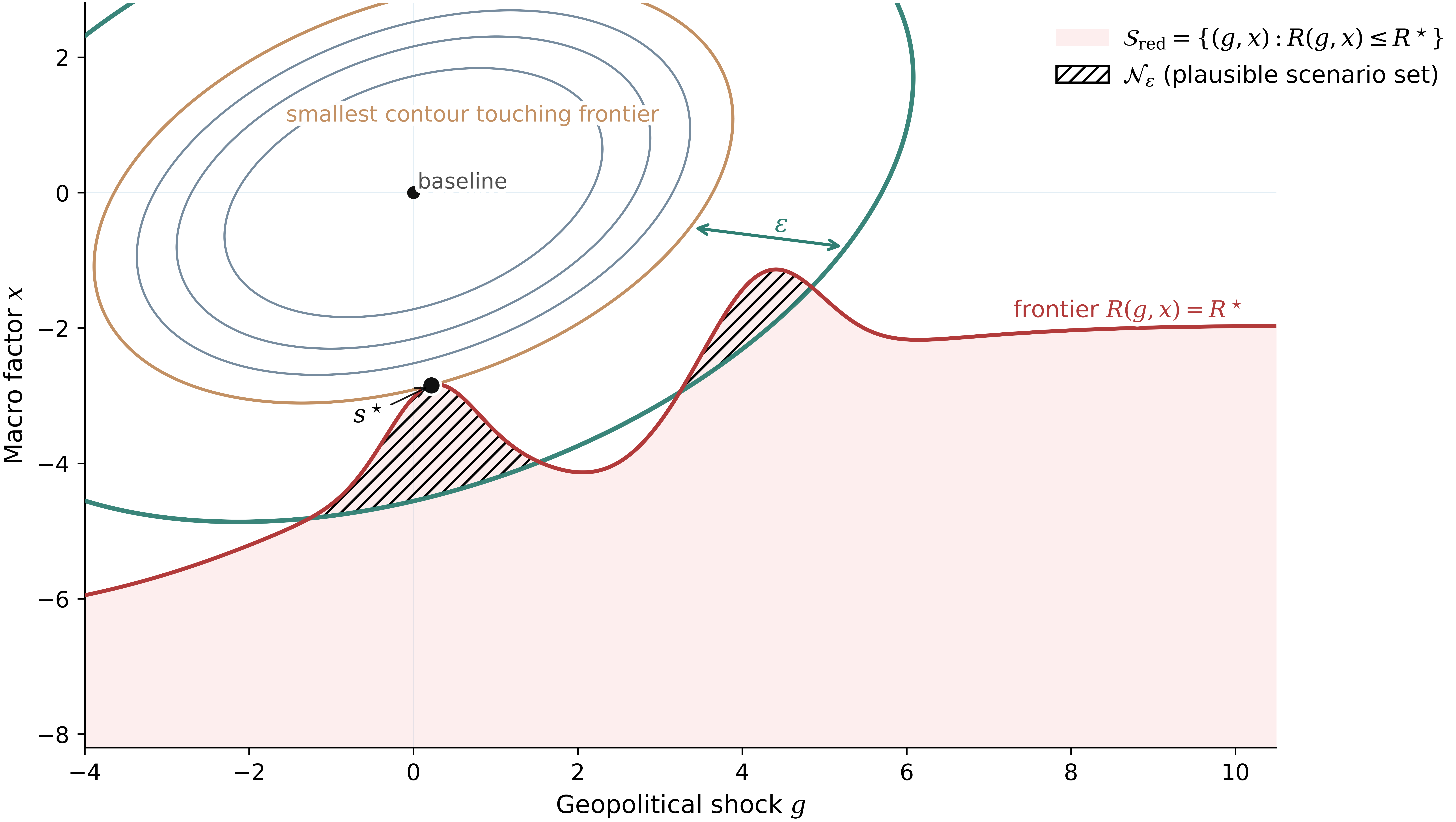}
  \caption{Geometry of the plausible reverse stress scenario set $\mathcal{N}_\varepsilon$ in the $(g,x)$ plane. The shaded region is the reverse stress breach scenario set (red zone) $\mathcal{S}_{\mathrm{red}}=\{(g,x):R(g,x)\le R^\star\}$ and the solid curve is the breakdown frontier $R(g,x)=R^\star$. Elliptic contours centred at $s=0$ are level sets of $d_\Sigma^{2}(s)=s^\top\Sigma^{-1}s$. The point $s^\star=(g^\star,x^\star)$ is the design point. The outer contour corresponds to $d_\Sigma^{2}(s)=d_\Sigma^{2}(s^\star)+\varepsilon$; the hatched region $\mathcal{N}_\varepsilon=\mathcal{S}_{\mathrm{red}}\cap\{s:\ d_\Sigma^2(s)\le d_\Sigma^2(s^\star)+\varepsilon\}$
  collects all capital-breaching scenarios within an $\varepsilon$ margin of the optimal log-likelihood level.}
  \label{fig:rst_near_optimal}
\end{figure}

\vspace{-0.5cm}
\paragraph{Plausibility score.}
Having established the plausibility score for the point reverse stress scenario, the same metric applies unchanged to the set-valued constructions considered here. The chi-squared calibration introduced in Section~\ref{sec:point_RST} applies unchanged to the set-valued constructions considered here. Accordingly, under normality assumption, any scenario $s$, in particular any $s\in\mathcal{S}_\eta$ or $s\in\mathcal{N}_\varepsilon$, is assigned the same tail-probability plausibility score
\begin{equation}
  \mathbb{P}\!\left(d_\Sigma^{2}(S)\ge d_\Sigma^{2}(s)\right)
  =
  1 - F_{\chi^{2}_{d}}\!\left(d_\Sigma^{2}(s)\right),
\end{equation}

The distinction between $\mathcal{S}_\eta$ and $\mathcal{N}_\varepsilon$ therefore lies solely in their membership constraints rather than in the plausibility metric itself. Membership in $\mathcal{S}_\eta$ additionally requires
$\delta(s,s^\star)\le\eta$ (or $d_\Sigma^2(s-s^\star)\le\eta$ under normality), reflecting proximity to the benchmark design point, whereas membership in $\mathcal{N}_\varepsilon$ requires
$d_\Sigma^2(s)\le d_\Sigma^2(s^\star)+\varepsilon$ (equivalently, $-\log f(s | \mathcal{I})\le -\log f(s^\star|\mathcal{I})+\varepsilon/2$),
capturing near-optimal plausibility irrespective of geometric distance from $s^\star$.

\subsection{Interpretation of the geopolitical reverse stress scenario}
\label{sec:interpretation}

Let $s^{\star}=(g^{\star},x^{\star})$ denote the solution to the geopolitical point reverse stress test. Evaluating the credit-risk transmission mappings at $s^{\star}$ yields the stressed exposure-level credit-risk parameters
\begin{equation}
  \text{PD}_i^{\star}=\text{PD}_i(g^{\star},x^{\star}),
  \qquad
  \text{LGD}_i^{\star}=\text{LGD}_i(g^{\star},x^{\star}),
  \qquad i=1,\dots,n,
\end{equation}
as well as the corresponding values of the portfolio loss quantile and the CET1 ratio,
\begin{equation}
  L_q^{\star}=L_q(g^{\star},x^{\star}),
  \qquad
  R(g^{\star},x^{\star})=R^{\star},
\end{equation}
where the latter equality holds under standard regularity conditions ensuring that the capital
constraint binds at the optimum.

The scalar $g^{\star}$ represents the geopolitical coordinate of the most probable configuration, under the reference joint distribution, that brings the portfolio to the boundary of the reverse stress breach scenario set. Conditional on this geopolitical stress level, the macro-financial vector $x^{\star}$ is selected endogenously as the most plausible companion shock, given the dependence structure given by $P_S(\,\cdot\,\vert\mathcal{I})$, that suffices to trigger the prescribed capital depletion. Accordingly, the benchmark scenario $s^{\star}$ should not be interpreted as a worst-case outcome. Rather, it corresponds to the most credible joint geopolitical and macro-financial realisation, relative to the baseline and historical variability, that breaches the capital threshold $R^\star$.

The economic interpretation of the reverse stress scenario follows directly from the stressed exposure-level credit-risk parameters. These exposure-level effects can be aggregated along relevant dimensions, such as economic sectors. Let $k(i)\in\{1,\dots,K\}$ denote the sector to which exposure $i$ belongs, and define the sector-level exposure weights by
\begin{equation}
  \omega_{i\mid k}
  :=
  \frac{\mathrm{EAD}_i}{\sum_{j:\,k(j)=k} \mathrm{EAD}_j},
  \qquad i:\,k(i)=k.
\end{equation}
Sector-level probabilities of default and losses given default are then defined as
exposure-weighted averages of the corresponding exposure-level quantities,
\begin{equation}
  \text{PD}_k(g,x)
  :=
  \sum_{i:\,k(i)=k} \omega_{i\mid k}\,\text{PD}_i(g,x),
  \qquad
  \text{LGD}_k(g,x)
  :=
  \sum_{i:\,k(i)=k} \omega_{i\mid k}\,\text{LGD}_i(g,x).
\end{equation}
Evaluating the sector-level aggregation mappings at the point reverse stress scenario $s^{\star}=(g^{\star},x^{\star})$ yields the stressed sectoral credit-risk parameters
\begin{equation}
  \text{PD}_k^{\star}
  :=
  \text{PD}_k(g^{\star},x^{\star}),
  \qquad
  \text{LGD}_k^{\star}
  :=
  \text{LGD}_k(g^{\star},x^{\star}),
  \qquad k=1,\dots,K.
\end{equation}
The collection $\left(\text{PD}_k^{\star},\text{LGD}_k^{\star}\right)_{k=1}^{K}$ provides a compact summary of how the geopolitical reverse stress scenario propagates across
economic sectors. This cross-sectional profile highlights concentration effects and heterogeneous sectoral
sensitivities to geopolitical risk. In particular, it clarifies whether the capital breach is primarily driven by a deterioration in default probabilities, by an increase in loss severities, or by a joint amplification of both
channels in specific sectors. By aggregating exposure-level responses into sector-level indicators, this representation
translates the abstract scenario coordinates $(g^{\star},x^{\star})$ into an economically interpretable
map of sectoral vulnerabilities, facilitating communication, risk diagnosis, and supervisory assessment. Such sector-level diagnostics are particularly useful for identifying clusters of exposure that warrant closer scrutiny under geopolitical stress and for designing targeted risk mitigation measures.

Finally, this interpretation extends beyond the point estimate $s^{\star}$ to the families of
plausible reverse stress scenarios introduced in Section~\ref{sec:multiple_scenarios}.
For any scenario $s=(g,x)\in\mathcal{S}_{\eta}$ or $s\in\mathcal{N}_{\varepsilon}$, scenario-specific
PDs and LGDs, as well as portfolio-level loss, can be evaluated, namely
\begin{equation}
  \text{PD}_i(g,x),\quad
  \text{LGD}_i(g,x),\quad
  L_q(g,x),\quad
  R(g,x),\quad
  d_{\Sigma}^{2}(s).
\end{equation}
This yields a family of severe but plausible reverse stress scenarios, associated with distinct
stressed outcomes and explicitly grounded in a transparent plausibility metric.

\section{Implementation of the reverse stress test}
\label{Section:Implementation}

This section outlines the operational framework for translating the theoretical model into a practical reverse stress testing exercise. It is structured around four components: an execution roadmap, the measurement of geopolitical risk using exogenous indices, the numerical implementation and optimization strategy, and a scenario selection methodology that reduces infinite admissible sets to a finite, decision-relevant set.

\subsection{Implementation Steps}

In practice, the reverse stress test framework can be implemented in four steps. First, the modeller selects a geopolitical risk indicator and a set of macro-financial factors, thereby defining the scenario vector $s=(g,x)^{\top}$, and estimates the reference covariance matrix $\Sigma$ from historical data; when relevant, this estimation can be adjusted for structural breaks, regime shifts, or other sources of non-stationarity. 
Second, the transmission mappings $(g,x)\mapsto \text{PD}$ and $(g,x)\mapsto \text{LGD}$ are calibrated, using either exposure-level information or sector-level aggregates, depending on data availability and modelling granularity. 
Third, the remaining building blocks of the capital framework are specified consistently with internal methodologies, including the dependence structure for defaults and any approximation used for risk-weighted assets. 
Fourth, the reverse-stress programme is solved numerically with a standard constrained optimiser (e.g.\ SQP or interior-point methods; \citealp{boyd2004convex}), yielding the least-unlikely capital-breaching scenario $s^{\star}$ and, when needed, the associated families $\mathcal{S}_{\eta}$ and $\mathcal{N}_{\varepsilon}$.

\begin{figure}[!htbp]
  \centering
  \begin{tikzpicture}[
    font=\small,
    box/.style={
      draw,
      rounded corners=2mm,
      align=left,
      inner sep=4mm,
      text width=0.88\linewidth
    },
    arrow/.style={-{Stealth[length=2.3mm]}, thick}
  ]

    \node[box] (b1) {%
      \textbf{Step 1: Scenario definition and reference model}\\ \smallskip
      Choice of geopolitical indicator and macro-financial factors; estimation of the reference covariance structure.
    };

    \node[box, below=10mm of b1] (b2) {%
      \textbf{Step 2: Calibration of transmission mappings}\\ \smallskip
      Specification of PD and LGD sensitivities to geopolitical and macro-financial shocks.
    };

    \node[box, below=10mm of b2] (b3) {%
      \textbf{Step 3: Specification of the capital framework}\\ \smallskip
      Modelling of portfolio aggregation, default dependence, and risk-weighted assets.
    };

    \node[box, below=10mm of b3] (b4) {%
      \textbf{Step 4: Numerical solution of the reverse stress test}\\ \smallskip
      Constrained optimisation yielding the point reverse stress test scenario and the associated sets.
    };

    \draw[arrow] (b1.south) -- (b2.north);
    \draw[arrow] (b2.south) -- (b3.north);
    \draw[arrow] (b3.south) -- (b4.north);

  \end{tikzpicture}
  \caption{Implementation roadmap of the geopolitical reverse stress testing framework.}
  \label{fig:implementation_roadmap}
\end{figure}
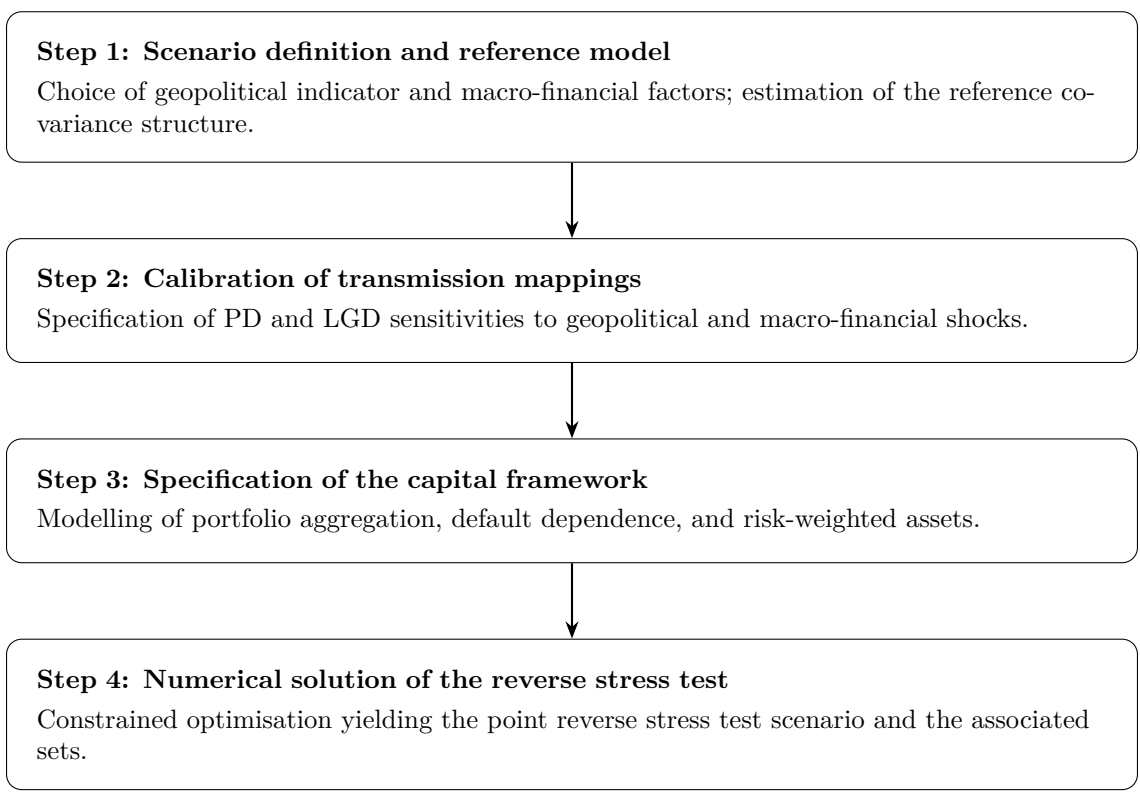

\subsection{Measurement and operationalization of geopolitical risk}
\label{Subsection: Geopolitical_index}

In the theoretical framework, the geopolitical component $G$ is modeled as an abstract scalar random variable capturing exogenous geopolitical shocks. In the empirical implementation, this latent component must be proxied by an observable indicator that provides a time-varying measure of geopolitical tensions while remaining plausibly exogenous to domestic macroeconomic and financial conditions.

A natural and widely used candidate is the Geopolitical Risk Index (GPR) developed by \citet{caldara2022geopolitical}. 
The GPR is constructed from the frequency of newspaper articles published in major international media outlets that report on geopolitical events related to wars, military tensions, terrorism, and international conflicts. 
By relying on automated text analysis of a consistent set of sources, the index delivers a transparent and replicable measure of geopolitical risk with broad international coverage and a long historical span. A key feature of the GPR is its decomposition into two complementary components. The Geopolitical Threats index captures rising tensions, military buildups, and hostile rhetoric that increase uncertainty without necessarily leading to immediate conflict. The Geopolitical Acts index, by contrast, reflects the realization or escalation of adverse geopolitical events, such as armed conflicts or major terrorist attacks. This distinction is particularly relevant in the context of stress testing and reverse stress testing. 
Geopolitical threats may affect expectations, investment decisions, risk premia, and credit conditions even in the absence of realized conflict, while geopolitical acts are more likely to generate abrupt, nonlinear adjustments in macro-financial variables.

Beyond its conceptual appeal, the GPR exhibits two properties that make it especially suitable for stress-testing applications. First, it is constructed independently of macroeconomic and financial variables, which supports its interpretation as an exogenous risk driver rather than an endogenous response to economic conditions. Second, its availability at a high frequency and over a long historical window allows both for the identification of typical geopolitical shocks and for the replay of historically observed episodes in scenario design.

While the GPR serves as our baseline proxy for the geopolitical component $G$, the framework developed in this paper is not tied to a specific indicator. Alternative measures of geopolitical risk can be accommodated provided they can be interpreted as exogenous shocks and embedded into the joint reference distribution of the scenario vector. 
These alternatives include machine-coded, news-based indicators derived from the GDELT database \citep{shawon2024assessing}, as well as proprietary indices developed by asset managers, such as Amundi’s Geopolitical Sentiment Tracker or BlackRock’s Geopolitical Risk Indicator, among many others. 

More generally, the geopolitical component $G$ may be defined at different levels of aggregation, depending on the scope of the analysis. Global indices are appropriate for assessing system-wide exposures, while regional or country-specific measures may be preferable for institutions with concentrated geographical portfolios. The stress-testing and reverse stress-testing methodology proposed in this paper remains agnostic with respect to this choice and can be applied to any geopolitical risk proxy that satisfies these minimal requirements.

\subsection{Numerical implementation and optimisation strategy}
\label{sec:num_impl}

From a computational standpoint, three ingredients are particularly important.

First, the numerical performance of the solver is greatly improved when the mapping
$s \mapsto R(s)$ is continuous and sufficiently smooth over the admissible domain.
In particular, when bounds on credit-risk parameters are imposed through hard truncations,
it is preferable to replace them with smooth saturating transformations, which preserve the
intended ranges while avoiding kinks that can impair convergence.
\medskip

Second, let $\Sigma = LL^\top$ and define the change of variables $y = L^{-1}s$
(equivalently, $s = Ly$). Under this reparameterisation, the plausibility penalty becomes
\[
\frac{1}{2} d_\Sigma^2(s)
=
\frac{1}{2} s^\top \Sigma^{-1} s
=
\frac{1}{2}\lVert y \rVert_2^2,
\]
so that the quadratic form is spherical in $y$-space: level sets are Euclidean spheres
rather than correlated ellipsoids in the original scenario space.\footnote{In whitened
coordinates, all components are on the same unitless scale: a one-unit move in any
coordinate increases the plausibility penalty in the same way. Equivalently, whitening
removes heterogeneous scaling and cross-correlation, improving conditioning and stabilising
step sizes in SQP or interior-point algorithms.}
Operationally, all distances (including those used to quantify ``diversity'') are computed in
$y$-space, while the capital constraint is evaluated as $R(Ly)\le R^\star$.
\medskip

Lastly, when the reverse stress event set $\mathcal{S}_{\mathrm{red}}$ is non-convex, or when the mapping
$R(\cdot)$ is highly non-linear, the constrained optimisation problem may admit multiple local optima.
To mitigate this issue, we adopt a multi-start strategy, running the solver from multiple
initialisations, such as random draws in the whitened space $y$ or a coarse grid over the geopolitical
dimension combined with conditional optimisation over the macro-financial factors. Among the resulting
feasible solutions, we retain the one with the smallest $d_\Sigma^2(\cdot)$.
This approach also facilitates the identification of alternative near-optimal scenarios, which can then
be ranked consistently according to their plausibility level.

\subsection{From feasible set to a finite scenarios list}
\label{subsec:scenario_list}

The set-valued outcomes introduced in Section~\ref{sec:multiple_scenarios}, the local neighbourhood $\mathcal{S}_\eta$ and the near-optimal set $\mathcal{N}_\varepsilon$, are infinite by construction: many combinations of geopolitical and macro-financial shocks can satisfy the breakdown condition $R(s)\le R^\star$.
In practice, however, reverse stress testing is a governance exercise and requires a \emph{finite} set of scenarios that can be discussed, challenged, and documented. We therefore complement the design point with a simple and auditable three-stage methodology that turns the relevant admissible set into a short list of $P$ representative scenarios. The stages are described below.

\vspace{-0.5cm}
\paragraph{Target set notation.}
Throughout this subsection, we denote by $\mathcal{E}$ the admissible set from which scenarios are to be selected:
\begin{equation}
\label{eq:target_set_E}
\mathcal{E}\in\{\mathcal{S}_\eta,\mathcal{N}_\varepsilon\},
\qquad
\mathcal{E}\subseteq \mathcal{S}_{\mathrm{red}}
\quad\text{with}\quad
\mathcal{S}_{\mathrm{red}}=\{s:\,R(s)\le R^\star\}.
\end{equation}
Recall that $\mathcal{S}_\eta$ is a local neighbourhood around the design point (e.g.\ $d_\Sigma^2(s-s^\star)\le\eta$)
and $\mathcal{N}_\varepsilon$ is a near-optimal plausibility set (e.g.\ under normality,
$d_\Sigma^2(s)\le d_\Sigma^2(s^\star)+\varepsilon$), both intersected with the red zone. Whenever a statement
depends on the choice of $\mathcal{E}$, we make it explicit.

The key point is that the framework already provides everything needed for selection: (i) a model-consistent
measure of \emph{plausibility} (via $d_\Sigma$ under normality, or any monotone equivalent under elliptical
models), and (ii) clear \emph{membership tests} for admissibility (capital breach and plausibility
constraints). The remaining task is not to ``invent'' scenarios, but to select a small subset that (a)
covers distinct narratives and (b) remains close to the plausibility frontier. This ``\textit{generate many
candidates, then reduce them}'' strategy is standard in stress testing, appearing in likelihood-based
approaches \citep{Glasserman2015}, systematic design methods \citep{flood2015systematic}, and recent
reverse-stress applications that compress large scenario clouds into a few interpretable families
\citep{BaesSchaanning2023}, in line with the multi-scenario perspective emphasised in
\citet{aikman2024multiple_scenarios_ecbwp2941}.  We compute all distances in whitened coordinates $y=L^{-1}s$ (so that $d_\Sigma^2(s)=\|y\|_2^2$), as
introduced in Section~\ref{sec:num_impl}.

\vspace{-0.5cm}
\paragraph{\textbf{Stage 1: Build a large candidate pool inside $\mathcal{E}$.}}
We first construct a pool $\mathcal{C}_N=\{s^{(1)},\dots,s^{(N)}\}$ with $N\gg P$ inside the target set
$\mathcal{E}$. The purpose of $\mathcal{C}_N$ is to make the final list robust and governance-ready: in
non-convex settings, several disconnected capital-breaching narratives may coexist, and a short reportable
list is meaningful only if it is selected from a pool that already covers these alternatives. We therefore
combine global exploration (to reach distinct narratives) and local exploration (to populate each narrative
with admissible variations).

\smallskip
\noindent\emph{Global exploration via anchors.}
We compute a set of feasible \emph{anchors} in the red zone
$\mathcal{S}_{\mathrm{red}}$, typically located in the vicinity of the breakdown boundary $\partial\mathcal{S}_{\mathrm{red}}=\{s:\,R(s)=R^\star\}$, up to numerical tolerances. To this end, the point reverse stress programme is solved from multiple initialisations in the whitened space. Distinct locally optimal feasible solutions are then retained after deduplication based on the $\|\cdot\|_2$ distance.
These anchors are not reported as such; they serve as economically interpretable centres for different narratives and protect the procedure against missing remote or disconnected near-optimal regions. Intuitively, anchors provide a ``skeleton'' of
the admissible region before we densify locally.

\smallskip
\noindent\emph{Anchors on a geopolitical intensity ladder.}
In our application, the scenario vector naturally decomposes as $s=(g,x)^\top$, where $g$ measures
geopolitical stress intensity and $x$ collects macro-financial factors. In practice, committees often
reason in terms of an \emph{intensity ladder} (mild / moderate / severe geopolitics) and ask a concrete,
operational question:
\emph{``If geopolitical stress is at level $g$, what macro-financial co-movements would be needed for a
capital breach?''} We answer this question by building additional anchors on a coarse grid
$\{g_j\}_{j=1}^J$. For each $g_j$, we compute the least-unlikely macro-financial companion shock that still
breaches the capital threshold:
\begin{equation}
\label{eq:conditional_anchor}
x^\star(g_j)\in \arg\min_{x}\ \frac12 d_\Sigma^2(g_j,x)
\quad \text{s.t.}\quad R(g_j,x)\le R^\star.
\end{equation}
The interpretation is direct and reportable: conditional on geopolitical intensity $g_j$, $x^\star(g_j)$ is
the most plausible macro-financial configuration that still produces a breach. This grid-based construction
is useful for two reasons. First, it guarantees explicit coverage across geopolitical intensities, which
helps avoid blind spots when the final list is debated. Second, it provides a transparent diagnostic of the
\emph{trade-off} between geopolitical and macro-financial stress needed to trigger a breach: inspecting the
mapping $g\mapsto x^\star(g)$ highlights levels of $g$ that require implausible macro co-movements (or,
conversely, levels where a breach is feasible under relatively plausible macro conditions). When the target
set is the near-optimal set ($\mathcal{E}=\mathcal{N}_\varepsilon$), we retain only those grid-based anchors
that also satisfy the near-optimal plausibility constraint (e.g.\ under normality,
$d_\Sigma^2(g_j,x^\star(g_j))\le d_\Sigma^2(s^\star)+\varepsilon$), ensuring that intensity-indexed anchors
remain close to the plausibility frontier.

\smallskip
\noindent\emph{Local exploration around anchors.}
Around each anchor $s^{\mathrm{anc}}$, we generate additional candidates by sampling perturbations in
whitened space and accepting points that satisfy the relevant membership tests. Let
$y^{\mathrm{anc}}=L^{-1}s^{\mathrm{anc}}$. We draw directions $u$ uniformly on the unit sphere in
$\mathbb{R}^d$ and radii $r\ge0$ from a compact interval, and set
\begin{equation}
\label{eq:local_draw}
y = y^{\mathrm{anc}} + r u,
\qquad s = Ly.
\end{equation}
We accept $s$ whenever $s\in\mathcal{S}_{\mathrm{red}}$ and $s\in\mathcal{E}$. Concretely, this means:
\vspace{-0,2cm}
\begin{itemize}
\item if $\mathcal{E}=\mathcal{S}_\eta$, accept when $d_\Sigma^2(s-s^\star)\le \eta$;
\item if $\mathcal{E}=\mathcal{N}_\varepsilon$ under normality, accept when
$d_\Sigma^2(s)\le d_\Sigma^2(s^\star)+\varepsilon$;
\item under elliptical models, replace $d_\Sigma^2(\cdot)$ by the corresponding monotone plausibility score
$\pi(\cdot)$ and apply the analogous threshold test.
\end{itemize}
When acceptance is low because the admissible region $\mathcal{E}$ is thin or highly curved, we replace
simple rejection sampling by a random-walk exploration constrained to $\mathcal{E}$ (hit-and-run / slice-type
moves): one proposes a move in $y$-space and keeps it if and only if the resulting $s=Ly$ remains in
$\mathcal{E}$. Importantly, this only requires a membership oracle (a yes/no check), not derivatives of
$R(\cdot)$. Step~1 thus produces a pool $\mathcal{C}_N$ that captures multiple narratives (via anchors,
including the geopolitical intensity ladder) and within-narrative sensitivity (via local exploration).

\paragraph{\textbf{Stage 2: Reduce the pool to $P$ representative, non-redundant scenarios.}}
We then convert $\mathcal{C}_N$ into a reportable list $\mathcal{C}_P=\{s^{[1]},\dots,s^{[P]}\}$ with $P$
small (typically 5--12). Since all candidates are admissible by construction, the reduction problem is
primarily one of diversity: we want to avoid reporting near-duplicates and instead span distinct parts of
the admissible region in a metric consistent with the reference model. We implement diversity selection in
whitened space using a maximin (farthest-point) rule,\footnote{This reduction step is conceptually close to
clustering, but it is not a $k$-means procedure. $k$-means selects centroids that minimise within-cluster
squared distances and these centroids need not correspond to admissible scenarios. By contrast, the
farthest-point (maximin) rule selects \emph{existing} candidates and explicitly targets coverage/diversity
in the reference-model metric (whitened Euclidean distance). When a clustering-based reduction is preferred,
a natural alternative is $k$-medoids in whitened space, which returns representative \emph{scenarios}
(medoids) rather than synthetic centroids.} initialised at the design point, $s^{[1]}=s^\star$. Let
$y^{(i)}=L^{-1}s^{(i)}$ denote the pool in $y$-space. We iteratively add the candidate that maximises its
distance to the already selected set:
\begin{equation}
\label{eq:farthest_point}
s^{[p]}
\in
\arg\max_{s\in\mathcal{C}_N\setminus \{s^{[1]},\dots,s^{[p-1]}\}}
\min_{\ell\le p-1}
\big\|L^{-1}(s-s^{[\ell]})\big\|_2,
\qquad p=2,\dots,P.
\end{equation}
This rule is deterministic, parameter-free, and directly auditable: each newly added scenario is the one
that most increases coverage of $\mathcal{E}$ under the plausibility-consistent metric. Figure~\ref{fig:scenario_list}
provides a visual summary of the candidate-generation step and of the farthest-point reduction that produces
a short, non-redundant list.

\begin{figure}[!htbp]
  \centering
  \includegraphics[width=0.85\linewidth]{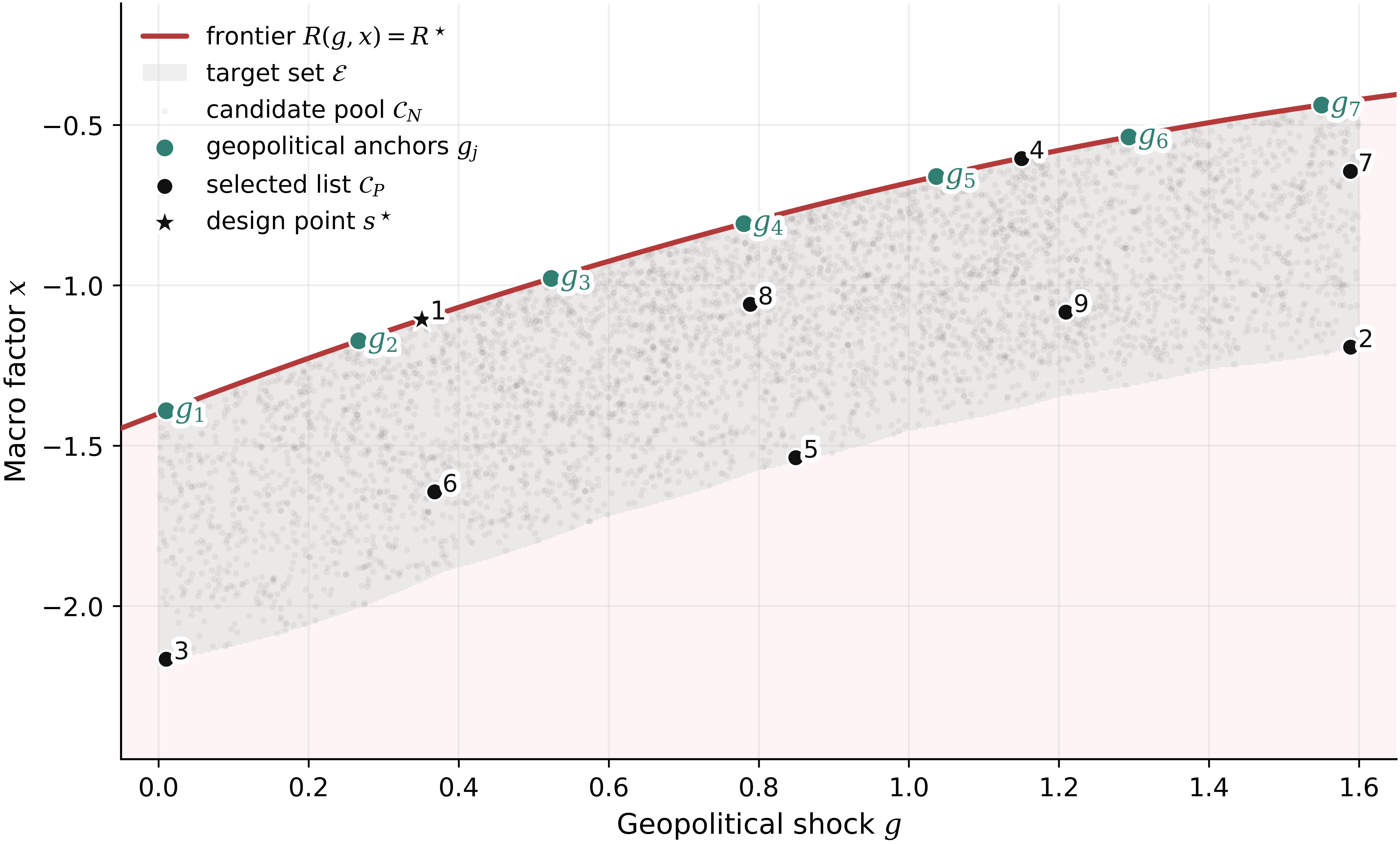}
  \caption{Illustration of the finite-scenario selection procedure in the $(g,x)$ plane.
  The red curve is the breakdown boundary $R(g,x)=R^\star$ and the shaded area is the target admissible set
  $\mathcal{E}$ (either $\mathcal{S}_\eta$ or $\mathcal{N}_\varepsilon$, depending on the object considered).
  Grey dots are admissible candidates in the pool $\mathcal{C}_N$, generated by global exploration around feasible
  anchors and local exploration in whitened space.
  Green markers $g_j$ indicate intensity-indexed anchors $(g_j,x^\star(g_j))$ defined by \eqref{eq:conditional_anchor}.
  Black markers are the reduced list $\mathcal{C}_P$ obtained by the farthest-point rule
  \eqref{eq:farthest_point}, initialised at the design point $s^\star$.}
  \label{fig:scenario_list}
\end{figure}

\paragraph{\textbf{Stage 3: Report severity, plausibility, and drivers in a standardised format.}}
For each selected scenario $s^{[p]}=(g^{[p]},x^{[p]})$, we report (i) severity via $R(s^{[p]})$, (ii)
plausibility via $d_\Sigma^2(s^{[p]})$ and the associated tail-probability score, and (iii) a driver
decomposition in whitened coordinates $y^{[p]}=L^{-1}s^{[p]}$. Narrative attribution follows directly from
$y^{[p]}$: we list the few coordinates with largest absolute values (together with their signs), which
identifies the dominant departures from baseline in a scale-free and model-consistent manner. Reporting
$g^{[p]}$ alongside these drivers makes the geopolitical intensity of each selected scenario explicit and
easy to compare across narratives. The resulting output is a finite, comparable, and auditable set of
severe-yet-plausible geopolitical reverse stress scenarios.

\section{Conclusion}
\label{Section:Conclusion}

This paper develops a formal and operational framework for reverse stress testing geopolitical risk in corporate credit portfolios. Starting from a prescribed capital breakdown outcome, we embed an explicit geopolitical risk component into a joint macro-financial scenario vector and map scenarios into stressed probabilities of default and losses given default. These stressed credit-risk parameters are propagated to portfolio tail losses through a tractable latent-factor
structure and translated into a stressed CET1 ratio by combining scenario-dependent losses and risk-weighted assets. Reverse stress testing is then formulated as a constrained optimisation problem that identifies the geopolitical point reverse stress test scenario or design point, defined as the
most likely macro-geopolitical scenario under which the capital constraint is violated, given a reference distribution for joint shocks. Under a Gaussian reference model, the geopolitical reverse stress test admits a transparent geometric interpretation. The point reverse stress scenario is characterised as the capital-breaching scenario that minimises the Mahalanobis distance to baseline conditions. This representation yields a scale-free plausibility score with a direct chi-squared calibration. Together, these properties provide a coherent probabilistic benchmark for assessing the severity of geopolitical stress scenarios and for comparing them across portfolios, model specifications, and historical episodes.

Beyond the identification of a single design point, the framework explicitly accommodates sets of plausible reverse stress outcomes. A local neighbourhood around the point reverse stress scenario captures uncertainty and supports sensitivity analysis in the vicinity of the most plausible capital breakdown. Complementarily, an $\varepsilon$-near-optimal reverse stress scenario set collects alternative capital-breaching scenarios whose plausibility remains close to that of the design point, even when these scenarios are distant in the scenario space. This set-valued construction is particularly relevant when the reverse stress event set is non-convex or when nonlinear transmission mechanisms generate multiple near-optimal narratives. Together, these objects address a central governance requirement of reverse stress testing: moving from a single internally coherent breaking scenario to a structured range of severe yet credible scenarios that can be ranked, discussed, and communicated in a consistent manner.

The proposed framework is intentionally modular. The scenario space can accommodate different geopolitical risk proxies and richer macro-financial representations, including alternative dependence structures or heavy-tailed reference distributions. The credit-risk block can be implemented either at the exposure level or at an aggregated sectoral level, and the capital mechanics can be aligned with internal methodologies for PD, LGD, and RWA. From an implementation perspective, the framework delivers outputs that are directly relevant for risk management and supervisory dialogue: a point reverse stress scenario anchored in an explicit geopolitical coordinate, an interpretable macro-financial companion shock, a decomposition of the capital breach through stressed PD and LGD channels, and a plausibility score that can be benchmarked against historical geopolitical episodes.

Several extensions would further enrich the framework. A natural direction is to allow for regime changes and tail dependence in the reference distribution, better reflecting the discrete and state-dependent nature of geopolitical events. Another extension is to broaden the outcome constraint to joint solvency and liquidity conditions, in line with supervisory interest in funding stress and operational resilience. Finally, a comprehensive empirical application estimating the transmission parameters and confronting the implied geopolitical severity $g^{\star}$ with historical and event-based benchmarks would provide an additional validation layer and clarify the practical range of plausible geopolitical reverse stress scenarios across institutions and portfolios.

\baselineskip=1 \normalbaselineskip

\bibliographystyle{apalike}
\bibliography{reverse_stress}

\newpage

\appendix

\section{Appendix}

\subsection{Multivariate Student scenarios}
\label{Appendix:student_reference}

This appendix provides a heavy-tailed alternative to the Gaussian reference model used in the main text.
All admissible sets (capital-breaching constraints and any additional feasibility restrictions) are unchanged; only the \emph{reference} distribution used to quantify plausibility is modified.
\\\\
Let $s=(g,x)^\top\in\mathbb{R}^d$ denote the scenario vector.
Assume that
\begin{equation}
  s \sim t_\nu(0,\Sigma),
  \label{eq:t_reference}
\end{equation}
with degrees of freedom $\nu>0$ and scatter matrix $\Sigma\succ0$.
The corresponding density is
\begin{equation}
  p_\nu(s)
  =
  c_{\nu,d}\,|\Sigma|^{-1/2}\left(1+\frac{1}{\nu}\,d_\Sigma^2(s)\right)^{-\frac{\nu+d}{2}},
  \qquad
  d_\Sigma^2(s):=s^\top\Sigma^{-1}s,
  \label{eq:t_density}
\end{equation}
where $c_{\nu,d}=\frac{\Gamma\!\left(\frac{\nu+d}{2}\right)}{\Gamma\!\left(\frac{\nu}{2}\right)(\nu\pi)^{d/2}}$.
Up to an additive constant (irrelevant for optimisation over $s$ when $\Sigma$ is fixed), the Student plausibility penalty is the negative log-density
\begin{equation}
  \ell_\nu(s)
  :=
  \frac{\nu+d}{2}\,
  \log\!\left(1+\frac{1}{\nu}\,d_\Sigma^2(s)\right).
  \label{eq:t_neglog}
\end{equation}
Let $\mathcal{S}_{red}$ denote the feasible set defined in Section~\ref{sec:rst_optim}, i.e.\ the set of scenarios satisfying the breakdown constraint and any additional restrictions.
Under the Student $t$ reference model, the least-unlikely reverse-stress scenario is obtained by
\begin{equation}
  \min_{s\in\mathcal{S}_{red}} \ \ell_\nu(s).
  \label{eq:rst_student}
\end{equation}
Since the mapping $u\mapsto \log(1+u/\nu)$ is strictly increasing on $[0,\infty)$ and the prefactor $(\nu+d)/2$ is positive, $\ell_\nu(s)$ is a strictly increasing transformation of $d_\Sigma^2(s)$.
Therefore,
\begin{equation}
  \arg\min_{s\in\mathcal{S}_{red}} \ \ell_\nu(s)
  \;=\;
  \arg\min_{s\in\mathcal{S}_{red}} \ d_\Sigma^2(s),
  \label{eq:argmin_equiv_student}
\end{equation}
which coincides with the Gaussian benchmark in the main text (where the objective is proportional to $d_\Sigma^2(s)$).
Hence, for a fixed scatter matrix $\Sigma$, switching from a Gaussian to a Student $t$ reference distribution leaves the optimisation problem unchanged in terms of its minimiser(s): the optimal reverse-stress scenario $s^\star$ is the same, and only its \emph{probabilistic interpretation} changes.
\\
To calibrate plausibility levels probabilistically, we need the distribution of $d_\Sigma^2(s)$ under \eqref{eq:t_reference}.
By the standard scale-mixture representation of the multivariate Student $t$ distribution (e.g.\ \citet[Ch.~5]{casella2002statistical}), let

\[
Z\sim \mathcal{N}(0,\Sigma),
\qquad
W\sim \chi^2_\nu,
\qquad
Z\ \perp\ W,
\]
and define
\[
S:=\sqrt{\frac{\nu}{W}}\,Z.
\]
Then $S\sim t_\nu(0,\Sigma)$.
Moreover,
\[
d_\Sigma^2(S)
=
S^\top\Sigma^{-1}S
=
\frac{\nu}{W}\,Z^\top\Sigma^{-1}Z.
\]
Since $Z^\top\Sigma^{-1}Z\sim \chi^2_d$ and is independent of $W$, we obtain
\begin{equation}
  \frac{1}{d}\,d_\Sigma^2(S)
  =
  \frac{Z^\top\Sigma^{-1}Z/d}{W/\nu}
  \ \sim\ F_{d,\nu},
  \label{eq:t_mahalanobis_F}
\end{equation}
i.e.\ the squared Mahalanobis distance (rescaled by $d$) follows a Fisher distribution with $(d,\nu)$ degrees of freedom.
\\\\
All neighbourhood constructions introduced in the main text (e.g.\ $\eta$-balls around $s^\star$ and $\varepsilon$-near-optimal sets) remain unchanged under a Student $t$ reference model.
Indeed, these sets are defined as intersections between the capital-breaching region $\mathcal{S}_{\mathrm{red}}$ and level sets of the Mahalanobis distance $d_\Sigma^2(\cdot)$, and therefore depend only on the feasible region and on thresholds for $d_\Sigma^2$.
What changes under a Student $t$ reference is solely the probabilistic calibration: the mapping from a given value of $d_\Sigma^2$ to a credibility level, through \eqref{eq:t_mahalanobis_F}.
\\\\
The plausibility of the optimal scenario $s^\star$ is summarised by the level $d_\Sigma^2(s^\star)$.
Under Gaussianity, $d_\Sigma^2(s)\sim\chi^2_d$, and the tail probability
\[
\mathbb{P}\!\left(d_\Sigma^2(s)\ge d_\Sigma^2(s^\star)\right)
=
1-F_{\chi^2_d}\!\left(d_\Sigma^2(s^\star)\right),
\]
where $F_{\chi^2_d}$ denotes the cumulative distribution function of $\chi^2_d$, provides a scale-free plausibility score.
\\\\
Under the Student $t$ reference model with $\nu$ degrees of freedom, one has $\frac{1}{d}d_\Sigma^2(s)\sim F_{d,\nu}$, and therefore
\[
\mathbb{P}\!\left(d_\Sigma^2(s)\ge d_\Sigma^2(s^\star)\right)
=
\mathbb{P}\!\left(F_{d,\nu}\ge \frac{d_\Sigma^2(s^\star)}{d}\right)
=
1-F_{F_{d,\nu}}\!\left(\frac{d_\Sigma^2(s^\star)}{d}\right),
\]
where $F_{F_{d,\nu}}$ denotes the cumulative distribution function of the Fisher distribution with $(d,\nu)$ degrees of freedom.

As illustrated in Fig.~\ref{fig:gauss_vs_student_samecov}, the choice of reference distribution affects the \emph{tail} plausibility attached to distant scenarios, even when the covariance matrix is held fixed. In particular, for a given Mahalanobis radius $d_\Sigma^2(s)=s^\top\Sigma^{-1}s$, a Student-$t$ reference (matched to the same covariance) assigns larger upper-tail probabilities $\pi(s)=\mathbb{P}(d_\Sigma^2(S)\ge d_\Sigma^2(s))$, hence smaller rarity $-\log_{10}\pi(s)$ in the extremes.

\begin{figure}[!ht]
  \centering
  \includegraphics[width=\linewidth]{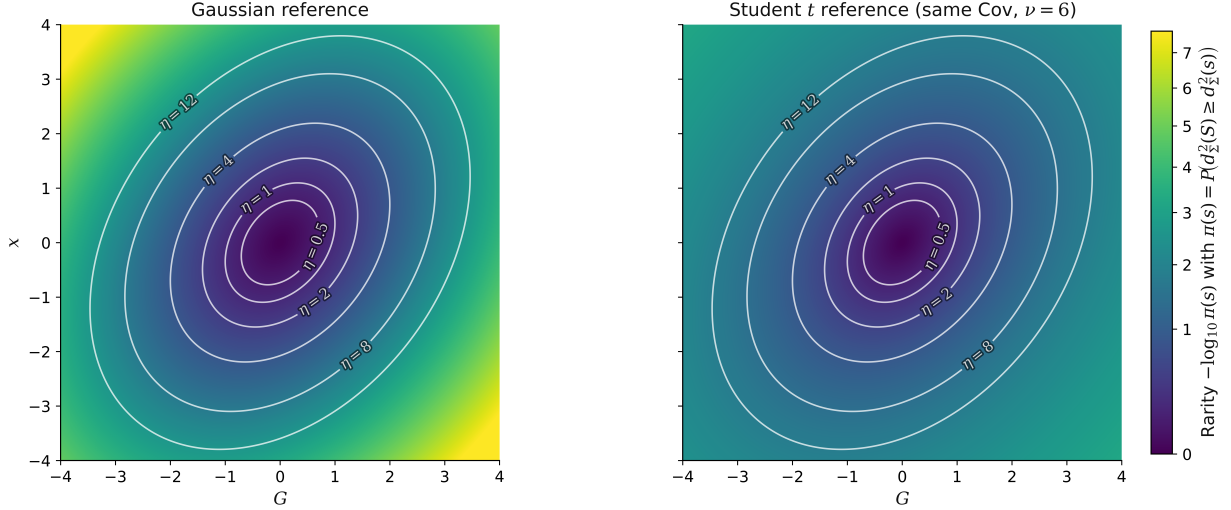}
  \caption{Gaussian vs.\ Student-$t$ reference plausibility on a common colour scale (scenario $s=(g,x)^\top$). Colours show rarity $-\log_{10}\pi(s)$ with $\pi(s)=\mathbb{P}(d_\Sigma^2(S)\ge d_\Sigma^2(s))$ and $d_\Sigma^2(s)=s^\top\Sigma^{-1}s$. The Student-$t$ panel uses $\nu=6$ and is rescaled to match the same covariance $\mathrm{Cov}(S)=\Sigma$. Contours are iso-levels of $d_\Sigma^2(s)=\eta$.}
  \label{fig:gauss_vs_student_samecov}
\end{figure}

\subsection{Sector-level reverse stress test with geopolitical transmission channels}
\label{Appendix:Sectoral_RST}

For practical implementation, it is often convenient to work directly with sector-level aggregates rather than exposure-level credit-risk parameters. Suppose that the portfolio is partitioned into $K$ sectors. For each sector $k$, we define the exposure at default $\text{EAD}_k$, baseline
credit-risk parameters $\text{PD}_k^0$ and $\text{LGD}_k^0$, and the asset correlation parameter $\rho_k$.
We further define sector exposure weights as
\begin{equation}
  w_k
  :=
  \frac{\text{EAD}_k}{\sum_{j=1}^{K} \text{EAD}_j}.
\end{equation}

The scenario vector remains $s=(g,x)^{\top}$, but the mapping from geopolitical and macro-financial conditions to credit-risk parameters is specified directly at the sector level.
A sector-level logit specification for probabilities of default is given by
\begin{equation}
  \log\!\left(
  \frac{\text{PD}_k(g,x)}{1-\text{PD}_k(g,x)}
  \right)
  =
  \log\!\left(
  \frac{\text{PD}_k^0}{1-\text{PD}_k^0}
  \right)
  + b_k^{\top} x + d_k g,
\end{equation}
while losses given default are modelled as
\begin{equation}
  \text{LGD}_k(g,x)
  =
  \text{LGD}_k^0 + c_k^{\top} x + e_k g.
\end{equation}

The coefficients $d_k$ and $e_k$ capture the direct transmission of geopolitical risk into sector-level default probabilities and loss severities, respectively. Sectors that are more exposed to geopolitical tensions—such as defence, energy, transportation, or export-intensive manufacturing—can therefore be assigned larger values of $|d_k|$ or $|e_k|$, reflecting higher sensitivity to geopolitical shocks.

Under this sectoral specification, an approximation to the portfolio loss quantile takes the form
\begin{equation}
  L_q(g,x)
  \approx
  \sum_{k=1}^{K}
  \text{EAD}_k\,\text{LGD}_k(g,x)\,
  \Phi\!\left(
  \frac{
    \Phi^{-1}\!\big(\text{PD}_k(g,x)\big)
    + \sqrt{\rho_k}\,\Phi^{-1}(q)
  }{
    \sqrt{1-\rho_k}
  }
  \right).
\end{equation}

A similar approximation can be employed for risk-weighted assets,
\begin{equation}
  \text{RWA}(g,x)
  =
  \sum_{k=1}^{K}
  \text{EAD}_k\,\text{RW}_k(g,x),
\end{equation}
where sector-level risk weights are given by
\begin{equation}
  \text{RW}_k(g,x)
  =
  \text{LGD}_k(g,x)
  \left[
  \Phi\!\left(
  \frac{
    \Phi^{-1}\!\big(\text{PD}_k(g,x)\big)
    + \sqrt{\rho_k}\,\Phi^{-1}(q)
  }{
    \sqrt{1-\rho_k}
  }
  \right)
  - \text{PD}_k(g,x)
  \right].
\end{equation}
When appropriate, this expression may be further approximated linearly as
\begin{equation}
  \text{RW}_k(g,x)
  =
  \alpha_k
  + \beta_k^{\text{RW}}
  \big(\text{PD}_k(g,x)-\text{PD}_k^0\big).
\end{equation}

The reverse stress testing problem is unchanged in form.
The objective remains to minimise $\tfrac{1}{2}d_{\Sigma}^{2}(s)$ over the set $\{s:\,R(s)\le R^{\star}\}$, where the CET1 ratio $R(s)$ is now computed from the sector-level mappings $(g,x)\mapsto(\text{PD}_k,\text{LGD}_k)$. In this formulation, the coefficients $d_k$ and $e_k$ provide an explicit and transparent parametrisation of the transmission of geopolitical risk into sectoral credit losses.

\end{document}